\documentclass[acmsmall, nonacm, authorversion]{acmart}

\AtBeginDocument{%
  }

\usepackage{macros}

\begin{document}


\title{On the Complexity of Hazard-Free Formulas}

\author{Leah {London Arazi}}
\orcid{0009-0006-3867-5478}
\affiliation{
  \institution{Tel Aviv University}
  \city{Tel Aviv}
  \country{Israel}
}
\email{leahl@mail.tau.ac.il}
\authornote{This research was co-funded by the European Union (ERC, EACTP, 101142020), the Israel Science Foundation (grant number 514/20) and the Len Blavatnik and the Blavatnik Family Foundation. Views and opinions expressed are however those of the author(s) only and do not necessarily reflect those of the European Union or the European Research Council Executive Agency. Neither the European Union nor the granting authority can be held responsible for them.}

\author{Amir {Shpilka}}
\orcid{0000-0003-2384-425X}
\affiliation{
  \institution{Tel Aviv University}
  \city{Tel Aviv}
  \country{Israel}
}
\email{shpilka@tauex.tau.ac.il}
\authornotemark[1]

\renewcommand{\shortauthors}{Leah London Arazi and Amir Shpilka}

\begin{abstract}
This paper studies the \emph{hazard-free formula complexity} of Boolean functions.

Our first result shows that unate functions are the only Boolean functions for which the monotone formula complexity of the hazard-derivative equals the hazard-free formula complexity of the function itself. Consequently, they are the only functions for which the hazard-derivative approach of Ikenmeyer et al. (J. ACM, 2019) yields optimal bounds. 

Our second result proves that the hazard-free formula complexity of a uniformly random Boolean function is at most $2^{(1+o(1))n}$. Prior to this, no better upper bound than $O(3^n)$ was known. 
Notably, unlike in the general case of Boolean circuits and formulas, where the typical complexity is derived from that of the multiplexer function with $n$-bit selector, the hazard-free formula complexity of a random function is smaller than the optimal hazard-free formula for the multiplexer by an exponential factor in $n$.

We provide two proofs of this fact. The first is direct, bounding the number of prime implicants of a random Boolean function and using this bound to construct a DNF of the claimed size. The second introduces a new and independently interesting result: a weak converse to the hazard-derivative lower bound method, which gives an upper bound on the hazard-free complexity of a function in terms of the monotone complexity of a subset of its hazard-derivatives. 

Additionally, we explore the hazard-free formula complexity of block composition of Boolean functions and obtain a result in the hazard-free setting that is analogous to a result of Karchmer, Raz, and Wigderson (Computational Complexity, 1995) in the monotone setting. 
We show that our result implies a stronger lower bound on the hazard-free formula depth of the block composition of the set covering function with the multiplexer function than the bound obtained via the hazard-derivative method.
\end{abstract}



\maketitle

\section{Introduction}  
\emph{Hazards} are a critical issue in the design of combinational circuits, which are key components in digital systems, realizing Boolean functions in the physical domain. A combinational circuit consists of gates connected by electrical wires. When setting voltage at the inputs, the signal propagates through the gates and wires until it reaches the output gate. This propagation involves inevitable signal \emph{delays}, which may vary on each path of the circuit. As a result, when the inputs change, some gates may respond faster than others, potentially leading to \emph{static hazards}. A static hazard occurs when a change in the input causes a brief, unintended fluctuation in the output, even though the output is expected to remain constant. Such hazards in combinational circuits can lead to unpredictable system behavior.

The study of hazards began in the late 1950s with the pioneering work of \cite{Huf57}, who suggested that it is possible to design circuits whose structure inherently prevents hazards without introducing delays. Huffman was the first to provide a hazard-free formula implementation for all Boolean functions, achieving a size of $O(\sqrt{n}\cdot 3^n)$. Later, additional works classified different types of hazards and suggested methods for identifying and eliminating them \cite{ung59, mcc62}. There is a broad body of literature on hazard-free computation, and we refer the interested reader to \cite{McC86, IKL+19} and references therein.

Recent research on hazard-free computation focuses on two main areas: electronic circuit design (e.g., \cite{FFL18, BLM20}) and computational complexity (e.g., \cite{IKL+19, KS20, Juk21, IK23}). While the motivation for the former is clear, the interest from the computational complexity community stems from the results of \cite{IKL+19}, which show that hazard-free formula complexity generalizes monotone formula complexity to all Boolean functions. This makes it an interesting restricted computational model of Boolean circuits and raises hope that hazard-free techniques may offer deeper insights into the general model.

Kleene's three-valued logic was introduced as an abstraction of static hazards \cite{yr64, Eic65}. This logic extends the standard Boolean values $\{0,1\}$ by introducing a third symbol, $\u$, representing an undefined or unstable value.

\begin{definition}[Kleene three-valued logic \cite{Kle52}]\label{def:k3}
Kleene's three-valued strong logic of indeterminacy extends the two-valued Boolean logic
by a third value $\u$. The Boolean values $\{0, 1\}$ are called \emph{``stable''} and $\u$ is called \emph{``unstable''}. 
The De Morgan gates are appropriately extended, as shown in \Cref{tbl:extended-gates}.
\end{definition}

\begin{table}[t]
\centering 
\caption{Extended De Morgan gates.}
\label{tbl:extended-gates}
\begin{tabular}[h]{c | c c c}
    or & $0$ & $\u$ & $1$ \\ 
    \hline
    $0$ & $0$ & $\u$ & $1$ \\
    $\u$ & $\u$ & $\u$ & $1$ \\
    $1$ & $1$ & $1$ & $1$ \\
\end{tabular}
\hspace{30pt}
\begin{tabular}[h]{ c | c  c  c }
    and & $0$ & $\u$ & $1$ \\ 
    \hline
    $0$ & $0$ & $0$ & $0$ \\
    $\u$ & $0$ & $\u$ & $\u$ \\
    $1$ & $0$ & $\u$ & $1$ \\
\end{tabular}
\hspace{30pt}
\begin{tabular}[h]{ c | c  c  c }
    not & $0$ & $\u$ & $1$ \\ 
    \hline
     & $1$ & $\u$ & $0$ \\
\end{tabular}
\end{table}

Now, any Boolean circuit or formula can be viewed as taking an input $\ternaryx{x}{n}$ and outputting a value in $\{0,\u,1\}$. The extended gates are monotone with respect to the following partial order:
\begin{definition}[$\precequ$ {\cite{Muk72}}]\label{def:bit-relations-precequ}
The relation $\precequ$ is a \emph{partial} order in which $\u \precequ 1$, $\u \precequ 0$ and $0,1$ are not comparable.
For $\ternaryx{x,y}{n}$, we denote  $x \precequ y$ if for every $i\in[n]$, $x_i \precequ y_i$.
\end{definition}

We present an abstraction of hazards using ternary algebra, following the definitions of \cite{FFL18,IKL+19}.

\begin{definition}[Resolution]\label{def:resolution} Let $\ternaryx{x}{n}$.
A resolution $\booleanx{y}{n}$ of $x$ is obtained by replacing every occurrence of $\u$ in $x$ by either $0$ or $1$. We denote by $R(x)$ the set of all resolutions of $x$, namely, $R(x):=\{\booleanx{y}{n} \mid x\precequ y\}$.
\end{definition}

\begin{definition}[Hazard]\label{def:hazard} 
Let $C$ be a Boolean circuit. $C$ has a \emph{hazard} at input $\ternaryx{x}{n}$ if and only if $C(x)=\u$ and there exists $\booleanx{b}{}$ such that for every $y \in R(x)$, it holds that $C(y)=b$.
\end{definition}

In words, there is a hazard at input $\ternaryx{x}{n}$ if the function is constant on every resolution of $x$ (i.e., regardless of how the unstable coordinates are resolved into stable values, the output remains unchanged), yet the circuit still outputs $\u$.
Put differently, a hazard-free circuit outputs $\u$ if and only if the output cannot be determined solely by the stable bits.

\begin{definition}[Hazard-free circuit]\label{def:hazard-free} 
Let $C$ be a Boolean circuit. $C$ is said to be \emph{hazard-free} if for every $\ternaryx{x}{n}$, $C$ does not have a hazard at $x$.
\end{definition}

By examining the truth tables in \Cref{tbl:extended-gates}, we observe that the extended De Morgan gates are hazard-free. However, circuits\footnote{An example of a Boolean circuit exhibiting hazards is shown in \cite[Figure 1]{IKL+19}.} and formulas composed of these gates do not necessarily retain the hazard-free property. Consider $F(x)=x\lor \neg x$, a Boolean formula over one bit. $F$ computes the constant Boolean function $1$, and by the truth tables in \Cref{tbl:extended-gates}, has a hazard at $x=\u$. 
Another less trivial example of a Boolean formula exhibiting hazards, is the formula in \Cref{fig:formula-with-hazard-example}. It is not difficult to verify that when evaluated on every resolution in $R(\u\u 00)$ the formula outputs $0$, but on $\u\u 00$, the formula outputs $\u$.

\begin{figure}
\centering
\begin{tikzpicture}
  [level distance=1cm, 
   level 1/.style={sibling distance=1.5cm},
   level 2/.style={sibling distance=2cm},
   level 3/.style={sibling distance=1cm},
   every node/.style={draw,circle,minimum size=8mm,fill=lightgray!50, inner sep=0pt},
   font=\small]

  \node{$\land$}
    child { node {$\neg x_1$} }
    child { node {$\land$}
      child { node {$\lor$}
        child { node {$x_3$} }
        child { node {$x_2$} }
      }
      child { node {$\lor$}
        child { node {$x_4$} }
        child { node {$\neg x_2$} }
      }
    };
\end{tikzpicture}
\caption{Formula with a hazard at $x_1=\u$, $x_2=\u$, $x_3=0$, $x_4=0$.}
\Description{An example of a formula that exhibits hazards. The depicted formula is $\neg x_1 \land ((x_3 \lor x_2) \land (x_4 \lor \neg x_2))$.}
\label{fig:formula-with-hazard-example}
\end{figure}

Not every ternary function $\ternaryf{f}{n}$ can be computed by a Boolean circuit. For instance, consider $f$ over one bit, defined such that $f(0)=f(1)=u$ and $f(u)=1$. The reason $f$ cannot be computed by a circuit is that every gate in \Cref{tbl:extended-gates} outputs a Boolean value when given a Boolean input and, as mentioned, is monotone with respect to $\precequ$. It can be proven by induction that circuits also preserve these properties.
Consequently, only a subset of all ternary functions can be computed by Boolean circuits. \cite{Muk72} refers to these functions as \emph{natural functions}.

\begin{definition}[Natural function]\label{def:natural-function} 
A function $\ternaryf{f}{n}$ is called a \emph{natural function} if it satisfies the following:
\begin{enumerate} 
    \item $f$ preserves stable values, i.e., for every Boolean input $\booleanx{x}{n}$ we have $\booleanx{f(x)}{}$.
    \item $f$ is monotone with respect to $\precequ$. Intuitively, replacing stable bits in the input with $\u$ can only cause the output to change to $\u$.
\end{enumerate}
\end{definition}

\begin{proposition}[{\cite[Theorem 3]{Muk72}}]\label{prop:natural-circuits}
A function $\ternaryf{f}{n}$ can be computed by a Boolean circuit if and only if $f$ is natural.
\end{proposition}

An equivalent way of defining a hazard-free circuit is by the ternary function it computes. A circuit is considered hazard-free if it calculates the following natural function:

\begin{definition}[The hazard-free extension of $f$]\label{def:hazard-free-extension} 
Let $\booleanf{f}{n}$ be a Boolean function. The \emph{hazard-free extension of $f$}, denoted $\ternaryf{\tilde{f}}{n}$, is defined as follows:
\[
    \tilde{f}(x)= 
    \begin{cases}
      0, & \text{if $f(y)=0$ for all $y\in R(x)$},\\
      1, & \text{if $f(y)=1$ for all $y\in R(x)$},\\
      \u, & otherwise.
    \end{cases}
\]
\end{definition}
\cite{IKL+19} showed a profound correspondence between monotone complexity and hazard-free complexity. Recall that a Boolean function is monotone if changing the value of an input coordinate from $0$ to $1$ cannot change the value of the function from $1$ to $0$. Monotone functions can be computed by monotone circuits, which are circuits that do not use negation gates. Using the  hazard-derivatives method, which we discuss next, they proved that for monotone functions, the hazard-free complexity and the monotone complexity are identical.

\begin{theorem}[Hazard-free and monotone circuit complexity are equivalent {\cite[Theorem 1.3]{IKL+19}}]\label{thm:monotone is hazard}
Let $\booleanf{f}{n}$ be a monotone Boolean function, then:
\[
\size_C^\u(f)=\size_C^+(f),
\]
where $\size_C^{+}(f)$ and $\size_{C}^{\u}(f)$ denote the minimal size of a monotone Boolean circuit  and a hazard-free circuit computing $f$, respectively.
\end{theorem}

In fact, \cite{Juk21} further proved that an optimal hazard-free circuit for a monotone function must be a monotone circuit. As superpolynomial lower bounds for monotone circuits \cite{Raz85, AB87, Tar88} are known, \cite{IKL+19} used \Cref{thm:monotone is hazard} to prove lower bounds on the hazard-free circuit complexity of monotone functions whose Boolean circuit complexity is only polynomial. 

While \Cref{thm:monotone is hazard} only concerns monotone functions, \cite{IKL+19} found a striking connection between hazard-free complexity and monotone complexity for all Boolean functions.  That is, the hazard-free complexity of every Boolean function $\booleanf{f}{n}$ is bounded from below by the monotone complexity of the \emph{hazard-derivative} of $f$, defined as follows:

\begin{restatable}[The hazard-derivative of a natural function {\cite[Definition 4.3]{IKL+19}}]{definition}{defhazardderivative}\label{def:hazard-derivative} 
Let $\ternaryx{f}{n}$ be a natural ternary function. The \emph{hazard-derivative} of $f$, denoted $\booleanf{\df{f}}{2n}$, is defined as follows:
\[
    \forall\booleanx{x,y}{n}, \quad \df{f}(x;y)=\begin{cases}
      0, & \text{if $f(x+\u y)=f(x)$},\\
      1, & \text{if $f(x+\u y)=\u$}.\\
    \end{cases}
\]
Where for every $i\in[n]$: 
\[
    (x+\u y)_{i}=\begin{cases}
      x_{i}, & \text{if $y_{i}=0$},\\
      \u, & \text{if $y_{i}=1$}.\\
    \end{cases}
\]
For a Boolean function $\booleanf{f}{n}$, we define $\df{f}:=\df{\tilde{f}}$. 
\end{restatable}

The $+$ operator in \Cref{def:hazard-derivative} also serves as the hazard-free XOR operator (see \Cref{def:xor}).\footnote{When considering $\u y_{i}=\begin{cases}
      0, & \text{if $y_{i}=0$},\\
      \u, & \text{if $y_{i}=1$}.\\
    \end{cases}$}

For a fixed $\booleanx{x}{n}$, $\df{f}(x;y)$ indicates whether ``perturbing'' $x$ according to $\u y$ causes $\tilde{f}$ to output an unstable value.
It is not difficult to see that $\dfunc{f}{x}$ is a monotone Boolean function, as if we perturb more bits in $x$, we are more likely to have two resolutions on which $f$ outputs different values. 

\begin{lemma}[The hazard-derivative is monotone {\cite[Lemma 4.6]{IKL+19}}]\label{lem:derivative-monotone}
Let $\booleanf{f}{n}$ be a Boolean function and $\booleanx{x}{n}$ be a fixed input, then $\dfunc{f}{x}$ is a \emph{monotone} Boolean function.
\end{lemma}

Every hazard-free circuit (or formula) for $f$ can be tweaked to yield a monotone circuit (or formula) for $\df{f}(x;\cdot)$ resulting in the following Theorem:\footnote{\cite[Theorem 4.9]{IKL+19} only speaks about circuit size, but the same proof holds for formulas as well.} 
\begin{restatable}[The hazard-derivative lower bound {\cite[Theorem 4.9]{IKL+19}}]{theorem}{thmderivativelowerbound}\label{thm:derivative-lower-bound}
Let $\booleanf{f}{n}$ be a Boolean function and $\booleanx{x}{n}$ be a fixed Boolean input, then: 
\[
\size_{C}^{+}(\dfunc{f}{x})  \le \size_{C}^{\u}(f) \quad \text{and}\quad
\size_{F}^{+}(\dfunc{f}{x})  \le \size_{F}^{\u}(f),
\]
where $\size_F^{+}(f)$ and $\size_{F}^{\u}(f)$ denote the minimal size of a monotone Boolean formula  and a hazard-free formula computing $f$, respectively.
\end{restatable}

\cite{IKL+19}, and later \cite{Juk21}, provided examples of non-monotone Boolean functions that have polynomial-sized Boolean circuits but whose hazard-derivatives require large monotone circuit complexity, thereby establishing strong lower bounds on the hazard-free circuit complexity of these functions. We note that until \cite{IKL+19}, hazard-free lower bounds were proved only for restricted models such as hazard-free depth $2$ circuits \cite{Eic65, ND95}. Therefore, their work motivates the study of the hazard-free model as a bridge between monotone complexity and Boolean circuit complexity.

Additionally, \cite[Theorem 1.3]{IKL+19} proved that the hazard-derivative method yields \emph{tight} bounds for monotone Boolean functions. That is, for a monotone function $\size_{C}^{+}(\dfunc{f}{x})=\size_{C}^{\u}(f)$ and $\size_{F}^{+}(\dfunc{f}{x})=\size_{F}^{\u}(f)$ when $x=\bar{0}$. This raises the following intriguing question.

\begin{question}\label{que:monotone-gap-criterion}
Is there a criterion that determines whether the hazard-derivative method yields an \emph{exact} lower bound? If so, is this criterion easy to verify?
\end{question}

\subsection{The monotone gap} 
\cite{IKL+19, Juk21} obtained exponential lower bounds on the hazard-free computation of many Boolean functions using \Cref{thm:derivative-lower-bound}. However, there are Boolean functions for which the hazard-derivative method fails to provide optimal lower bounds on hazard-free complexity due to the low monotone complexity of the hazard-derivatives. \cite{IK23} highlights this issue and refers to it as the \emph{monotone barrier}. They define breaking the monotone barrier as achieving stronger lower bounds than those obtained by the hazard-derivative method. 

As an example of a result that breaks the monotone barrier, \cite{IK23} note that the hazard-derivatives of the parity function on $n$ variables, $\xor_n$, are simply OR functions on $n$ literals, which have trivial formulas of size $n$. On the other hand, it is known that the De Morgan formula complexity of $\xor_n$ is $\Omega(n^2)$, while it is not hard to see that any formula for $\xor_n$ is hazard-free.

To quantify the gap between the hazard-free complexity of a given Boolean function and the monotone complexity of its hazard-derivatives, we define a measure, similar the one defined in \cite[Section $7$]{Juk21}, which we refer to as the \emph{monotone gap}.

\begin{definition}[The monotone gap]\label{def:the-monotone-gap}
Let $\booleanf{f}{n}$ be a Boolean function. The \emph{monotone gap of $f$} is defined as:
\[
\monogap(f):=\frac{ \size_F^\u(f)}{\max\limits_{\booleanx{x}{n}} \size_F^+(\dfunc{f}{x})}.
\]
\end{definition}

For example, the monotone gap of $\xor_n$ is $\Omega(n)$. As another example, in \Cref{subsec:andreev}, we prove that the monotone gap of Andreev's function \cite{And87} is $\Omega\left(\frac{n^2}{\log n}\right)$. In an attempt to understand the power of the hazard-derivative method (or its weaknesses), we wish to answer the following:
\begin{question}\label{que:monotone-gap}
Let $\booleanf{f}{n}$ be a Boolean function. How large can $\monogap(f)$ be? What is $\monogap(f)$ for a typical Boolean function?
\end{question}

\subsection{The hazard-free KW game} 
\cite{IK23} studied a generalization of the classic Karchmer-Wigderson game \cite{KW90} in the hazard-free setting (see \Cref{section:communication-complexity} for background on KW games). 

\begin{restatable}[Hazard-free KW game {\cite[Definition 5]{IK23}}]{definition}{defhazardfreekwgame}\label{def:hazard-free-kw-game}
Let $\booleanf{f}{n}$ be a Boolean function. The \emph{hazard-free Karchmer-Wigderson game} of $f$, denoted $\kwgame_{f}^{\u}$, is defined as follows: Alice receives $\ternaryx{x}{n}$ such that $\tilde{f}(x)=1$, Bob receives $\ternaryx{y}{n}$ such that $\tilde{f}(y)=0$, and their goal is to output a coordinate $i\in[n]$ such that $x_{i} \neq y_{i}$ and $x_{i},y_{i} \neq \u$.
\end{restatable}

As $\kwgame^{\u}_f\subseteq \tilde{f}^{-1}(1) \times \tilde{f}^{-1}(0) \times [n]$, and every $x\in \tilde{f}^{-1}(1)$ and $y\in \tilde{f}^{-1}(0)$ are valid inputs\footnote{There must exist a coordinate $i\in[n]$ such that $x_i\neq y_i$ and $x_i,y_i$ are stable. Otherwise, $x,y$ have a common resolution, in contradiction to \Cref{def:hazard-free-extension}.}, $\kwgame^{\u}_f$ has a potentially larger set of inputs than the classic game. As a result, the communication matrix may be more complex. Similarly to the classic KW game, the hazard-free KW game preserves a tight connection between the underlying communication problem complexity and the optimal \emph{hazard-free} formula depth and size. 

\begin{theorem}[{\cite[Theorem 7]{IK23}}]\label{thm:kw-hazard-free}
Let $\booleanf{f}{n}$ be a Boolean function. Then,
\[
    \text{$\depth_{F}^{\u}(f)=CC(\kwgame_{f}^{\u})$ and $\size_{F}^{\u}(f)=\monorect(\kwgame_{f}^{\u})$},
\]
where $\depth_{F}^{\u}(f)$ is the depth of the shallowest hazard-free formula computing $f$, $CC(\mathcal{R})$ is the communication complexity of the relation $\mathcal{R}$, and $\monorect(\mathcal{R})$ is the minimal number of leaves in a protocol that solves $\mathcal{R}$.
\end{theorem}

Like the monotone version of the KW game, \cite{IK23} proved that the communication complexity of $\kwgame_{f}^{\u}$ does not decrease even if it is played only on the \emph{prime implicants} of $f$ and the \emph{prime implicates} of $f$. 
\begin{theorem}[{\cite[Theorem 10]{IK23}}]\label{thm:kw-hazard-free-implicants}
For any Boolean function $\booleanf{f}{n}$, the communication complexity of $\kwgame_{f}^{\u}$ remains unchanged even if we restrict Alice's input to the prime implicants of $f$ and Bob's input to the prime implicates of $f$.
\end{theorem}
The notion of prime implicants and prime implicates is well established in Boolean algebra. Following \cite{IK23}, we provide analogous definitions using ternary algebra as previously defined in \cite{KS20}[p.3] and \cite{IK23}[p.7].

\begin{definition}[Implicants and implicates of a Boolean function]\label{def:implicant}
Let $\booleanf{f}{n}$ be a Boolean function. We define the implicants of $f$, denoted by $I^{(f)}_1$, and the implicates of $f$, denoted by $I^{(f)}_0$, as follows: 
\[
I^{(f)}_1:=\tilde{f}^{-1}(1), \quad I^{(f)}_0:=\tilde{f}^{-1}(0).
\]
\end{definition}

\begin{definition}[Prime implicants and prime implicates of a Boolean function]\label{def:prime-implicant}
Let $\booleanf{f}{n}$ be a Boolean function.
We say $\ternaryx{x}{n}$ is a \emph{prime implicant} of $f$ if the following holds:
\begin{itemize}
    \item $x$ is an implicant of $f$.
    \item For every stable bit $\booleanx{x_i}{}$ in x, we have $\tilde{f}(\replace{x}{i}{\u})=\u$, where $\replace{x}{i}{b}$ stands for replacing the value of the $i$'th coordinate in $x$ with $b\in\{0,\u, 1\}$.
\end{itemize}
A \emph{prime implicate} is defined analogously. We denote the set of prime implicants of $f$ by $\implicants{f}$ and the set of prime implicates of $f$ by $\implicates{f}$.
\end{definition}

For the remainder of the paper, we abuse notation and use $\kwgame^{\u}_f$ to denote the restricted game. We next define the communication matrix of the $\kwgame_f^\u$ game.

\begin{definition}
    [Communication matrix of $\kwgame_f^\u$]\label{def:communication-matrix-kwu}
Let $\mathcal{X}:=(p_1^{(1)},\dots,p_1^{(k)})$ and $\mathcal{Y}:=(p_0^{(1)},\dots,p_0^{(\ell)})$ be an arbitrary ordering of the elements in $\implicants{{f}}$ and $\implicates{f}$, respectively. We call $\mathcal{X}$ and $\mathcal{Y}$ the row labels and column labels, respectively. The communication matrix of $\kwgame_f^\u$, denoted $\matkwu{f}{}$, is defined such that every entry of $\matkwu{f}{}$ contains the set of all stable coordinates in which $p_1^{(i)}$ and $p_0^{(j)}$ differ:
\[
    \forall i\in [k], j\in [\ell], \quad (\matkwu{f}{})_{i,j}=\{z\in [n]: (p_1^{(i)} + p_0^{(j)})_{z}=1\}.
\]

For $\ternaryx{y}{n}$, we denote by $\replaceall{y}{s}{a}$ the result of replacing all occurrences of $\ternaryx{s}{}$ in $y$ with $\ternaryx{a}{}$. We then use $\replaceall{(p_1^{(i)} + p_0^{(j)})}{\u}{0}:=\{z\in [n]: (p_1^{(i)} + p_0^{(j)})_{z}=1\}$ as an abbreviated notation. 
\end{definition}

Another interesting observation is that, for a monotone function, the hazard-free KW game captures the monotone formula complexity.

\begin{theorem}[{\cite[Theorem 11]{IK23}}]\label{thm:kw-hazard-free-monotone}
Let $\booleanf{f}{n}$ be a monotone Boolean function. Then, $\kwgame_{f}^{\u}$ and $\kwgame_{f}^{+}$ are equivalent. 
\end{theorem}

Consequently, the size and depth of an optimal hazard-free formula are identical to those of its monotone counterpart. This provides a proof of \Cref{thm:monotone is hazard} for formulas instead of circuits. 

By using communication complexity methods to prove lower bounds on hazard-free formulas, \cite[Theorems 19, 23]{IK23} determined the exact hazard-free formula size of the multiplexer function.

\begin{definition}[The multiplexer function \cite{IK23}]\label{def:mux}
The function $\booleanf{\mux_n}{n+2^n}$, called the multiplexer function, is defined by
\[
\mux_n(s,x)=x_{\text{bin}(s)},
\]
where $\booleanx{s}{n}$, $\booleanx{x}{2^n}$ and $\text{bin}(s)$ is the natural number represented by $s$ in binary basis. 
\end{definition}

\begin{theorem}[Hazard-free formula size for the multiplexer function {\cite[Exact Bounds]{IK23}}]\label{thm:mux-exact-bound}
\[
\size^\u_F(\mux_n)=2\cdot 3^n-1.
\]
\end{theorem}
This result is much better than what can be achieved using the hazard-derivative method since the hazard-derivatives of the multiplexer function have quasi-linear monotone formula complexity (see \cite[Proposition 1]{IK23}). 

\subsection{Universal upper bounds}
An immediate consequence of \Cref{thm:mux-exact-bound} is that any Boolean function $f$ can be implemented by a hazard-free formula of size $2\cdot 3^n-1$. Indeed, it holds that:
\[
f(x) = \mux_n(x, f(0,0,\ldots,0), \ldots,f(1,1,\ldots,1)),
\] 
therefore, the hazard-free formula for the multiplexer, when hardwired the truth table of $f$, gives a hazard-free formula for $f$. This result is the first to break the  $O(\sqrt{n}\cdot 3^n)$  upper bound of Huffman \cite{Huf57}. 

We note that this state of knowledge is much poorer compared to what is known for hazard-free circuits. Analogous to \cite{Sha49, Mul56}, Jukna \cite[Section 7]{Juk21} proved that any Boolean function has a hazard-free circuit of size $O(2^n/n)$, applying the following:
\begin{proposition}[{\cite[Proposition 3]{Juk21}}]\label{prop:hazard-free-consensus-construction}
Let $C_0, C_1$ be arbitrary hazard-free circuits over $x_1,\dots,x_{n-1}$, and $x_n$ a new variable. Then, the circuit
\begin{equation}\label{eq:jukna-recursion}
    C := (\neg x_n\land C_0) \lor (x_n\land C_1) \lor (C_0 \land C_1)
\end{equation}
is hazard-free.
\end{proposition}
Therefore, naturally, we can construct a hazard-free circuit for $\booleanf{f}{n}$ by using hazard-free implementations $C_0$ and $C_1$ for $f_0(x_1,\dots,x_{n-1}):= f(x_1,\dots,x_{n-1}, 0)$ and $f_1(x_1,\dots,x_{n-1}):= f(x_1,\dots,x_{n-1}, 1)$, respectively. However, in the case of formulas, the size of the formula obtained by \eqref{eq:jukna-recursion} is $O(3^n)$, since it is composed of 3 hazard-free formulas over $n-1$ bits.\footnote{Note that the extra $(C_0 \land C_1)$ term is necessary. Say that for some $\ternaryx{x}{n-1}$ we have $\tilde{f}(x,\u)=1$. Hence, by definition, we also have $\tilde{f}_0(x)=1,\tilde{f}_1(x)=1$, and so $C_0(x)=1,C_1(x)=1$. But, $(\neg x_n\land C_0) \lor (x_n\land C_1)$ will evaluate to $\u$.}

This result matches the lower bound obtained via the same counting arguments used in the Boolean setting. In particular, for a typical Boolean function (i.e., one sampled uniformly at random from all Boolean functions on $n$ bits), we have $\size_C^{\u}(f)=\Theta(\size_C(f))$.

On the other hand, counting arguments give a lower bound of order $\Omega(2^n/\log n)$ for the size of the hazard-free formula of a random Boolean function \cite{RS42,Sha49}, whereas \Cref{thm:mux-exact-bound} only implies an $O(3^n)$ upper bound. Thus, there is an exponential gap between the two bounds. This naturally leads to the following question. 

\begin{question}\label{que:universal-ub}
    Denote with $\size^\u_F(n)$ the maximal hazard-free complexity of an $n$-bit Boolean function. What is the asymptotic value of $\size^\u_F(n)$? What is $\size^\u_F(f)$ of a typical Boolean function?
\end{question}

\subsection{The hazard-free KRW conjecture}
Karchmer, Raz and Wigderson \cite{KRW95} introduced the KRW conjecture as an approach for establishing super-polynomial lower bounds for De Morgan formulas. Proving this conjecture would imply $\text{P}\not\subseteq \text{NC}^1$. A detailed discussion can be found in \cite{Mei20}. 

To state the conjecture, we define the block-composition of two Boolean functions $\booleanf{f}{n}$ and $\booleanf{g}{m}$ as
\[
f\diamond g(X):=f(g(X_1),\dots,g(X_n)),
\]
where $X$ is a $n\times m$ Boolean matrix, and $\booleanx{X_i}{m}$ is the $i$'th row of $X$. For a Boolean function $f$, we denote with $\depth_F(f)$ the depth of the shallowest formula computing $f$.

\begin{conjecture}[The KRW conjecture \cite{KRW95}]\label{conj:krw}
Let $\booleanf{f}{n}$ and $\booleanf{g}{m}$ be non-constant Boolean functions. Then it holds that:\footnote{We use $\approx$ instead of $\geq$ as we can allow a less exact inequality, see \cite{KRW95, GMWW17}.} 
\[
\depth_F(f\diamond g)\approx \depth_F(f)+\depth_F(g).
\]
Following \Cref{thm:kw-classic}, the conjecture can be rephrased as:
\[
CC(\kwgame_{f\diamond g})\approx CC(\kwgame_f)+CC(\kwgame_g).
\]
\end{conjecture}

Roughly, the conjecture states that the smallest depth formula for computing $f\diamond g$ is obtained from the composition of the shallowest formula for $f$ and the shallowest formula for $g$.

We next state an analogous conjecture in the hazard-free setting. 
\begin{restatable}[Hazard-Free KRW conjecture]{conjecture}{conjhazardfreekrw}\label{conj:hazard-free-krw}
Let $\booleanf{f}{n}$ and $\booleanf{g}{m}$ be non-constant Boolean functions. Then the following holds:
\[
\depth_F^\u(f\diamond g)\approx \depth_F^\u(f)+\depth_F^\u(g).
\]
\end{restatable}

Similarly to \cite{DM18, dRMN+20}, we state a stronger version of the conjecture on the formula size.

\begin{restatable}[Hazard-free strong KRW conjecture]{conjecture}{conjhazardfreestrongkrw}\label{conj:hazard-free-strong-krw}
Let $\booleanf{f}{n}$ and $\booleanf{g}{m}$ be non-constant Boolean functions. Then the following holds:
\[
\size_F^\u(f\diamond g)\approx \size_F^\u(f)
\cdot \size_F^\u(g).
\]
\end{restatable}

The KRW conjecture has been extensively studied through various special cases and simplified communication problems, as intermediate steps toward proving \Cref{conj:krw}. The monotone version of the KRW conjecture is most related to the hazard-free setting due to the connection to the hazard-derivative, which is a monotone function.
\begin{conjecture}[Monotone KRW conjecture \cite{KRW95}] \label{conj:monotone-krw}
Let $\booleanf{f}{n}$ and $\booleanf{g}{m}$ be monotone Boolean functions. Then the following holds:
\[
\depth_F^+(f\diamond g) \approx \depth_F^+(f)+\depth_F^+(g).
\]
\end{conjecture}
In \cite[Theorem 3.1]{dRMN+20}, \Cref{conj:monotone-krw} was proved for monotone inner functions $g$ whose depth complexity can be lower-bounded by a lifting theorem. In an earlier work, \cite{KRW95} observed that the monotone KW game can be played on the prime implicants and implicates of the monotone function, aligning well with the results concerning the hazard-free KW game presented in \Cref{thm:kw-hazard-free-monotone}. Notably, the hazard-free KRW conjecture generalizes the monotone KRW conjecture, as proving the former for every $f$ and $g$ would also imply the latter.

\subsection{Our results}
\paragraph{Unateness as a criterion.}
We provide a complete answer to \Cref{que:monotone-gap-criterion} by proving that the hazard-derivative method yields exact lower bounds only for unate Boolean functions:
\begin{restatable}[Unateness as a criterion]{theorem}{thmkwfnotbreaksiff}\label{thm:kw-f-not-breaks-if}
Let $\booleanf{f}{n}$ be a non-constant Boolean function. Then, there exists $\booleanx{x}{n}$ such that:
\[
    \size_{F}^{+}(\dfunc{f}{x}) = \size_{F}^{\u}(f),
\]
if and only if $f$ is a unate function.
\end{restatable}
A unate Boolean function is defined as follows:
\begin{definition}[Unate function in $x_i$]\label{def:unate-function-xi}
A Boolean function $\booleanf{f}{n}$ is \emph{positive unate} in $x_i$ if for all $\booleanx{x}{n}$ the following holds:
\[
f(\replace{x}{i}{0}) \le f(\replace{x}{i}{1}),
\]
Similarly, a function is \emph{negative unate} in $x_i$ if for all $\booleanx{x}{n}$ the following holds:
\[
f(\replace{x}{i}{0}) \ge f(\replace{x}{i}{1}).
\]
We say $f$ is \emph{unate in $x_i$} if $f$ is either positive or negative unate in $x_i$.
\end{definition}

\begin{definition}[Unate function]\label{def:unate-function}
A Boolean function $\booleanf{f}{n}$ is \emph{unate} if it is unate in $x_1,\dots,x_n$.
\end{definition}

Monotone functions are a subset of unate functions, which  increase or decrease monotonically in each variable. Given the truth table of a Boolean function as input, we can efficiently determine whether it is unate or not by inspecting the truth table. We prove this result by analyzing the hazard-free KW game of the function and its hazard-derivatives.

\paragraph{A universal formula upper bound for most Boolean functions.} 
We prove that the hazard-free formula size of a random Boolean function is $2^{(1+o(1))n}$, thereby partially answering \Cref{que:universal-ub}. This gives an exponential improvement over the best known $O(3^n)$ upper bound of \cite{IK23}. It is important to note that the upper bound of \cite{IK23} holds for all functions, so it is conceivable that some functions require hazard-free formula size larger than $2^{(1+o(1))n}$. Even if this is the case, then neither counting arguments nor the hazard-derivative method, which can only provide a lower bound of at most $2^n$, will be able to prove it.

\begin{restatable}[A universal upper bound for most Boolean functions]{theorem}{universalub}
    \label{thm:universal-ub}
For a random Boolean function $\booleanf{f}{n}$, we get that with high probability:
\[
\frac{2^n}{\log n}\leq \size_{F}^\u(f)  \leq 2^n\cdot n^{(1-o(1))\log n}.
\]
\end{restatable}
Thus, we demonstrate that, unlike Boolean circuits, Boolean formulas and hazard-free circuits -- where the typical complexity matches that of the multiplexer function -- the hazard-free formula complexity is smaller than the optimal hazard-free formula for the multiplexer function by an exponential factor in $n$.

We provide two proofs for \Cref{thm:universal-ub} (see Sections \ref{subsec:random-function-dnf} and \ref{subsec:hazard-derivative-upper-bound}). The second relies on two other results that are interesting on their own. The first shows that, with high probability, the hazard-derivatives of a random Boolean function do not have too large monotone formulas. In particular, this provides a strong answer to \Cref{que:monotone-gap}, showing that the monotone gap can be exponential in the number of variables.

\begin{restatable}{proposition}{proprandomfunctiongap}\label{prop:random-function-gap}
Let $\booleanf{f}{n}$ be a random Boolean function. Then, with high probability,
\[
\monogap(f)=\omega\left(\frac{2^n}{ n^{\log n}}\right).
\]
\end{restatable}

The second ingredient is an \emph{upper bound} using hazard-derivatives. For $b\in\{0,1\}$, we say that $p_{b}\in \implic{f}{b}$ \emph{derives} $x$ if $p_{b}\precequ x$. We denote the subset of $\implic{f}{b}$ containing the strings that derive $x$ by $\implic{f}{b}|_x$.

\begin{restatable}[Hazard-derivatives upper bound]{theorem}{thmderivativesupperbound}\label{thm:derivative-upper-bound}
Let $\booleanf{f}{n}$ and $b\in\{0,1\}$. Let $X_b\subseteq f^{-1}(b)$. Then, 
\[
\text{if}\quad \implic{f}{b}=\bigcup_{x\in X_b}\implic{f}{b}|_x,
\quad \text{then it holds that}\quad 
\size_{F}^\u(f)\le \sum_{x\in X_b}\size_{F}^{+}(\dfunc{f}{x}).
\]
\end{restatable}
We also provide an example in \Cref{subsection:tight-lower-bound-range}  that shows that \Cref{thm:derivative-upper-bound} yields tight bounds for range functions (functions that accept all inputs whose weight is within some range).

\paragraph{Composition of hazard-free functions.}
Following \cite{KRW95}, we establish a lower bound on the hazard-free complexity of block-composition of functions and demonstrate how it can be used to surpass the hazard-derivative method.
\begin{restatable}[Lower bound via direct sum]{proposition}{proplowerboundviadirectsum}\label{prop:lower-bound-direct-sum}
Let $\booleanf{f}{n}$ be a monotone Boolean function. Let $\booleanf{g}{m}$ a Boolean function. Let $\Phi:X\times Y \rightarrow \{0,1\}$ and $\Psi:X'\times Y'\rightarrow \{0,1\}$ be \emph{functions} such that $\Phi$ reduces to $\kwgame_f^{+}$ and $\Psi$ reduces to $\kwgame_g^\u$. Then, the following holds:
\[
CC(\kwgame^\u_{f\diamond g})\ge\log(\rank(M_\Phi))+\log(\rank(M_\Psi)).
\]
\end{restatable}

\Cref{prop:lower-bound-direct-sum} provides a step toward advancing hazard-free lower bound techniques, and we hope it encourages further research in this direction. As an application, we prove a lower bound on the depth of the composition of the set covering function and the multiplexer function (see \Cref{sec:composition-hazard-free} for definitions).

\begin{restatable}[Composition of a set covering and the multiplexer function]{proposition}{propcompmonotonemux}
    \label{prop-intro:comp-monotone-mux}
Let $\mux_m$ be the multiplexer function and $MC_{k,n}$ be the set covering function, such that $k=c\log n$ for some suitable constant $c>0$. Then,
\[
\depth_{F}^\u(MC_{k,n} \diamond \mux_m) \ge \log(3^m)+\log\binom{n}{c \log n} \ge\log(3)\cdot m +\Omega(c(\log n)^2).
\]
\end{restatable}

\subsection{Organization of the paper}
In \Cref{section:preliminaries}, we present some basic definitions, as well as basic properties of the implicates and implicants of Boolean functions, with a focus on monotone and unate functions. Additionally, we provide background on communication complexity and prove the simplification lemma (\Cref{lem:simplification-lemma}). In \Cref{sec:kw-games-of-the-derivative}, we study the KW hazard-free game of the hazard-derivative.
In \Cref{sec:tightness-of-hdm}, we prove our first result (\Cref{thm:kw-f-not-breaks-if}), and in \Cref{subsec:andreev}, we present an example of an explicit function that exhibits a large monotone gap. \Cref{sec:universal} contains the proofs of the universal upper bound for most Boolean functions (\Cref{thm:universal-ub}), as well as 
the proof that a random Boolean function exhibits a large monotone gap (\Cref{prop:random-function-gap}), and the hazard-derivatives upper bound result (\Cref{thm:derivative-upper-bound}). In \Cref{sec:composition-hazard-free}, we study the composition of Boolean functions in the hazard-free setting and prove \Cref{prop:lower-bound-direct-sum}. We also demonstrate its applicability by proving \Cref{prop-intro:comp-monotone-mux}. 
We conclude by presenting some open problems in \Cref{sec:open}.

\section{Preliminaries}\label{section:preliminaries}
For an integer $n$, we set $[n]:=\{1,\ldots,n\}$. For $\booleanx{x}{n}$, we define $\neg x$ to be the Boolean input such that for every $i\in[n]$, it holds that $(\neg x)_i=\neg x_i$. For $\booleanx{x,y}{n}$, we denote $x \le y$ if for every $i\in[n]$, $x_i \le y_i$ (as integers).\footnote{Note the difference with \Cref{def:bit-relations-precequ} in which $0$ and $1$ are incomparable, whereas here $\u$ is not included in the ordering.} A Boolean function $\booleanf{f}{n}$ is \emph{monotone} if for every $\booleanx{x,y}{n}$ such that $x \le y$ it holds that $f(x) \le f(y)$.

For $s_1\in\{0,\u,1\}^n$ and $s_2\in\{0,\u,1\}^m$, we denote as $s_1\cdot s_2$ the concatenation of $s_1$ and $s_2$. For $\ternaryx{x}{n}$, we define the \emph{$b$-weight} of $x$ as $|x|_b:=|\{i\in[n]: x_{i}=b\}|$ for some $b\in\{0,\u,1\}$. 

Let $\ternaryf{g}{m}$ be a natural function, and let $\ternaryx{X}{nm}$ be a ternary $n\times m$ matrix. Then, we denote $g[X]:=(g(X_1),\dots, g(X_n))$ and $\df{g}[X;Y]:=(\df{g}(X_1;Y_1),\dots,\df{g}(X_n;Y_n))$, where $X_i$ is the $i$'th row of $X$. 

A (De Morgan) Boolean circuit is a directed acyclic graph in which nodes of in-degree $0$ are labeled with input variables, and each internal node is labeled with one of the Boolean operators $\{\land,\vee,\neg\}$, where the first two have fan-in $2$ and the third has fan-in $1$. Nodes of fan-out $0$ are called outputs. An output gate computes a Boolean function in the natural way. A (De Morgan) formula is a circuit whose graph is a tree. As in \Cref{def:hazard-free}, we say that a formula is hazard-free if it has no hazards.

A circuit (formula) is monotone if it does not involve $\neg$ gates. Each monotone circuit (formula) computes a monotone function, and each monotone function can be computed by a monotone circuit (formula). From now on, we speak only of De Morgan circuits and formulas, omitting the term ``De Morgan''.

The \emph{size} of a Boolean circuit $C$, denoted $\size(C)$, is defined as the number of vertices and edges in the underlying graph. The \emph{depth} of a Boolean circuit $C$, denoted $\depth(C)$, is defined as the length of the longest directed path from an input to an output gate. The size of a formula $F$, denoted $\size(F)$, is defined as the number of leaves in the underlying tree.

We denote the minimal size of a circuit computing $f$ by $\size_{C}(f)$, and the minimal depth of a formula computing $f$ by $\depth_{F}(f)$. We add a superscript from $\{\u,+\}$ to the notation whenever we work in the hazard-free ($\u$) or the monotone ($+$) setting, respectively.
For example, $\depth_{F}^{\u}(f)$ is the minimal depth of a hazard-free formula computing $f$.

Throughout the paper we use the hazard-free XOR operation, defined as follows:
\begin{definition}[XOR operation]\label{def:xor} We use \emph{$+$} to denote the $\xor$ operator in the hazard-free setting. Let $\ternaryx{b, s}{}$, then:
\[
    b+s:=
    \begin{cases}
      0, & \text{$b=s$ and $b,s \neq \u$},\\
      1, & \text{$b\neq s$ and $b,s \neq \u$},\\
      \u, & otherwise.
    \end{cases}
\]
Computing $x+y$ for $\ternaryx{x,y}{n}$ is done coordinatewise.
\end{definition}

\subsection{Implicants and implicates}\label{sec:implicants}
We prove the a useful property of $\implicants{f}$ and $\implicates{f}$ for a general Boolean function $f$. Later, we focus on monotone and unate functions. 

\begin{proposition}\label{prop:implicant-and-neg-implicate}
Let $\booleanf{f}{n}$ be a Boolean function. Let $p_1\in\implicants{f}$ and $p_0\in \implicates{f}$. Let $i,j\in[n]$ such that $(p_1)_i\neq \u$ and $(p_0)_j\neq \u$. Then,
\begin{enumerate}
    \item There exists $q_0\in\implicates{f}$ such that $i$ is the only different and stable coordinate of $q_0$ and $p_1$.\label{item:implicant-and-neg-implicate-1}
    \item There exists $q_1\in\implicants{f}$ such that $j$ is the only different and stable coordinate of $q_1$ and $p_0$.\label{item:implicant-and-neg-implicate-2}
\end{enumerate} 
\end{proposition}
\begin{proof}
We prove \eqref{item:implicant-and-neg-implicate-1}, the proof of \eqref{item:implicant-and-neg-implicate-2} is interchangeable. Denote $b:=(p_1)_i\in \{0,1\}$. Consider $p_1'=\replace{p_1}{i}{\u}$. Since $p_1$ is a prime implicant, we have $\tilde{f}(p'_1)=\u$ and therefore there exists $z\in R(p_1')$ such that $f(z)=0$.
Let $q_0\in \implicates{f}$ such that $z\in R(q_0)$. 
As $p_1'\precequ z$ and $q_0\precequ z$, $z$ agrees with $p_1$ and $q_0$ on every stable coordinate except $i$. We must have $(q_0)_i=\neg (p_1)_i=\neg b$, as otherwise $p_1,q_0$ would have a common resolution. Therefore, the $i$'th coordinate is the only different and stable coordinate.
\end{proof}

\subsubsection{Monotone Boolean functions}\label{subsec:prime-implicants-monotone}
The implicants and implicates of a monotone Boolean function can be studied by considering only the maximal and minimal resolutions with respect to $\le$:
\begin{definition}[The maximal and the minimal resolutions w.r.t $\le$]\label{def:max-and-min-res}
Let $\ternaryx{x}{n}$. We call $\hat{x}:=\replaceall{x}{\u}{1}$ the \emph{maximal} resolution of $x$ with respect to $\le$ and $\check{x}:=\replaceall{x}{\u}{0}$ the \emph{minimal} resolution of $x$ with respect to $\le$.
\end{definition}

\begin{lemma}[Implicants and implicates of a monotone function]\label{lem:implicants-implicates-monotone}
Let $\booleanf{f}{n}$ be a monotone Boolean function. Let $\ternaryx{p}{n}$, then:
\begin{enumerate}
    \item $p\in I_1^{(f)}$ if and only if $f(\check{p})=1$.\label{item:implicants-implicates-monotone-1}
    \item $p\in I_0^{(f)}$ if and only if $f(\hat{p})=0$.\label{item:implicants-implicates-monotone-2}
\end{enumerate}
\end{lemma}
\begin{proof}
We prove \eqref{item:implicants-implicates-monotone-1}, the proof of \eqref{item:implicants-implicates-monotone-2} is analogous. $p\in I_{1}^{(f)}$ if and only if $\forall z\in R(p), f(z)=1$. Since $f$ is monotone w.r.t $\le$ and $\check{p}\in R(p)$ is minimal w.r.t $\le$, the later holds if and only if $f(\check{p})=1$. 
\end{proof}

\begin{lemma}[Prime implicants and prime implicates of a monotone function]\label{lem:prime-implicants-implicates-monotone}
Let $\booleanf{f}{n}$ be a monotone Boolean function. Let $\ternaryx{p}{n}$ and denote $S:=\{i\in [n]: p_i\neq \u\}$, then:
\begin{enumerate}
    \item $p \in \implicants{f}$ if and only if $f(\check{p})=1$ and for every $i\in S$ it holds that $f(\replace{\check{p}}{i}{0})=0$.\label{item:prime-implicants-implicates-monotone-1}
    \item $p \in \implicates{f}$ if and only if $f(\hat{p})=0$ and for every $i\in S$ it holds that $f(\replace{\hat{p}}{i}{1})=1$.\label{item:prime-implicants-implicates-monotone-2}
\end{enumerate}
\end{lemma}
\begin{proof}
We prove \eqref{item:prime-implicants-implicates-monotone-1}, the proof of \eqref{item:prime-implicants-implicates-monotone-2} is similar. By \Cref{def:prime-implicant}, $p$ is a prime implicant of $f$ if and only if $p\in I_1^{(f)}$ and $\forall i\in S, \tilde{f}(\replace{p}{i}{\u})=\u$. By \Cref{lem:implicants-implicates-monotone}, we get $p\in I_1^{(f)}$ if and only if $f(\check{p})=1$, so it is left to prove $\forall i\in S, \tilde{f}(\replace{p}{i}{\u})=\u$ if and only if $\forall i\in S, f(\replace{\check{p}}{i}{0})=0$:

$\Rightarrow$ Let $i\in S$. Assume towards contradiction that $f(\replace{\check{p}}{i}{0})=1$. As $\replace{\check{p}}{i}{0}$ is the minimal resolution of $\replace{p}{i}{\u}$ w.r.t $\le$, it follows from 
\Cref{lem:prime-implicants-implicates-monotone} that $\tilde{f}(\replace{p}{i}{\u})=1$, a contradiction since $\tilde{f}(\replace{p}{i}{\u})=\u$.

$\Leftarrow$ Let $i\in S$. From our assumptions we have $f(\replace{\check{p}}{i}{0})=0$ and $f(\check{p})=1$. Since $\check{p}, \replace{\check{p}}{i}{0}\in R(\replace{p}{i}{\u})$ we get $\tilde{f}(\replace{p}{i}{\u})=\u$.
\end{proof}

\begin{corollary}\label{cor:prime-monotone}
Let $\booleanf{f}{n}$ be a monotone Boolean function, then:
\begin{enumerate}
    \item Every $\implicant \in \implicants{f}$ satisfies $\implicant\in\{1,\u\}^{n}$. \label{item:prime-monotone-1}
    \item Similarly, $\implicates{f} \subseteq  \{0,\u\}^{n}$. \label{item:prime-monotone-2}
\end{enumerate}
\end{corollary}
\begin{proof}
We prove \eqref{item:prime-monotone-1}, the proof of \eqref{item:prime-monotone-2} is similar. By \Cref{lem:prime-implicants-implicates-monotone}, if there exists a coordinate $i\in [n]$ such that $(p_1)_i=0$ then $\check{p_1}=\replace{\check{p_1}}{i}{0}$, a contradiction, since $f(\check{p_1})=1$ and $f(\replace{\check{p_1}}{i}{0})=0$.
\end{proof}

\subsubsection{Unate Boolean functions}
The next lemmata generalize \Cref{cor:prime-monotone}. 

\begin{restatable}[The prime implicants of a unate Boolean function in $x_i$]{lemma}{funateimplicants}\label{lem:f-unate-implicants}
Let $\booleanf{f}{n}$ be a Boolean function. Then:
\begin{enumerate}
    \item $f$ is \emph{positive} unate in variable $x_{i}$ if and only if there is no $\implicant \in \implicants{f}$ such that $(\implicant)_{i}=0$. \label{item:f-unate-implicants-1}
    \item $f$ is \emph{negative} unate in variable $x_{i}$ if and only if there is no $\implicant \in \implicants{f}$ such that $(\implicant)_{i}=1$.\label{item:f-unate-implicants-2}
\end{enumerate}    
\end{restatable}
\begin{proof}
We prove \eqref{item:f-unate-implicants-1}, as \eqref{item:f-unate-implicants-2} can be proved analogously.

$\Rightarrow$ Assume towards contradiction that there exists $\implicant \in \implicants{f}$ such that $(\implicant)_{i}=0$. Since $f$ is positive unate in $x_{i}$, then for every $z\in R(\implicant)$ it holds that:
\[
    1 = f(z)\le f(\replace{z}{i}{1}).
\]
Hence, $\replace{\implicant}{i}{\u}\in I_{1}^{(f)}$, in contradiction to our assumption that $\implicant$ is a \emph{prime} implicant of $f$.

$\Leftarrow$ Assume towards contradiction $f$ is not positive unate in $x_{i}$. Therefore, there exists $\booleanx{z}{n}$ such that:
\[
    0 = f(\replace{z}{i}{1}) < f(\replace{z}{i}{0}) = 1.
\]
Hence, there exist $\implicant \in \implicants{f}$ and $\implicate \in \implicates{f}$ such that $\implicant \precequ \replace{z}{i}{0}$ and $\implicate \precequ \replace{z}{i}{1}$. We must have $(\implicant)_{i}=0$, otherwise, $(\implicant)_{i}=\u$ and then we get $f(\replace{z}{i}{1}) = f(\replace{z}{i}{0})=1$.
\end{proof}

\begin{restatable}[The prime implicates of a unate Boolean function in $x_i$]{lemma}{funateimplicates}\label{lem:f-unate-implicates}
Let $\booleanf{f}{n}$ be a Boolean function. Then:
\begin{enumerate}
    \item $f$ is \emph{positive} unate in variable $x_{i}$ if and only if there is no $\implicate \in \implicates{f}$ such that $(\implicate)_{i}=1$. \label{item:f-unate-implicates-1}
    \item $f$ is \emph{negative} unate in variable $x_{i}$ if and only if there is no $\implicate \in \implicates{f}$ such that $(\implicate)_{i}=0$.\label{item:f-unate-implicates-2}
\end{enumerate}    
\end{restatable}

The proof of \Cref{lem:f-unate-implicates} is similar.

\subsection{Communication complexity}\label{section:communication-complexity}
In this section, we provide some basic notation and definitions from communication complexity. We assume familiarity with the basic concepts.
See \cite{KN97,RY20} for textbooks on communication complexity. 

Let $X,Y,Z$ be a finite sets. For a relation $\mathcal{R} \subseteq X \times Y \times Z$, we denote by $CC(\mathcal{R})$ the minimal communication complexity required to solve $\mathcal{R}$. 
Any protocol $P$ that solves a relation $\mathcal{R} \subseteq X \times Y \times Z$ induces a partition of $X \times Y$ into $\mathcal{R}$-monochromatic rectangles. The number of rectangles, denoted $\monorect(P)$, is the number of the leaves in $P$. We denote by $\monorect({\mathcal{R}})$ the minimum of $\monorect({P})$ among all protocols $P$ that solve $\mathcal{R}$. 
We shall abuse notation and use $CC(M)$ or $\monorect(M)$ instead of $CC({\mathcal{R}})$ or $\monorect({\mathcal{R}})$ when $M$ is a communication matrix of $\mathcal{R}$.

Let $\mathcal{R} \subseteq X \times Y \times Z$ be a relation and let a protocol $P$ for $\mathcal{R}$. Say there exists an input $(x^*, y^*)$ for which there is only one valid output, $z^* \in Z$. Since $P$ solves $\mathcal{R}$,  $(x^*,y^*)$ reaches a leaf in the protocol tree that is labeled with $z^*$. The associated rectangle is referred as a \emph{$z^*$-uniform rectangle}:

\begin{definition}[$z^*$-uniform combinatorial rectangle]\label{def:d-colored-combinatorial-rec}
Let $\mathcal{R}\subseteq X\times Y \times Z$ be a communication problem. Let $\rec$ be a \emph{$\mathcal{R}$-monochromatic rectangle}. We call $\rec$ a \emph{$z^*$-uniform} rectangle if there exists an entry $(x^*,y^*)\in \rec$ such that $\{z\in Z: (x^*,y^*,z)\in \mathcal{R}\}=\{z^*\}$.
\end{definition}

Let $\mathcal{R}\subseteq X\times Y \times Z$ and $\mathcal{R}'\subseteq X'\times Y'\times Z'$ be relations. We say that $\mathcal{R}$ is \emph{reducible} to $\mathcal{R}'$, and denote $\mathcal{R}\le \mathcal{R}'$, if there exist functions $\phi_1:X\rightarrow X'$, $\phi_2:Y\rightarrow Y'$ and $\psi:Z'\rightarrow Z$ such that for every $x\in X,y\in Y$ and $z'\in Z'$:
\[
(\phi_1(x),\phi_2(y),z')\in \mathcal{R}'\Rightarrow(x,y,\psi(z'))\in R.
\]
It is clear that if $\mathcal{R}\le \mathcal{R}'$ then $CC(\mathcal{R})\le CC(\mathcal{R}')$.

We also need the notion of direct sum of two relations.

\begin{definition}[Direct sum of relations {\cite[Definition 2]{KRW95}}]\label{def:direct-sum-relations}
Given two relations $\mathcal{R} \subseteq X \times Y \times Z$ and $\mathcal{R}' \subseteq X' \times Y' \times Z'$, we define the \emph{direct sum} of $\mathcal{R}$ and $\mathcal{R}'$ as:
\[
\mathcal{R} \otimes \mathcal{R}' \subseteq (X \times X') \times (Y \times Y') \times (Z \times Z'),
\]
where $((x_1, x_2), (y_1, y_2), (z_1, z_2)) \in \mathcal{R} \otimes \mathcal{R}'$ if and only if $(x_1, y_1, z_1) \in \mathcal{R}$ and $(x_2, y_2, z_2) \in \mathcal{R}'$.
\end{definition}

It is clear that if $\mathcal{S}\le\mathcal{R}$ and $\mathcal{S}'\le\mathcal{R}'$, then
$\mathcal{S} \otimes \mathcal{S}'\le\mathcal{R}\otimes \mathcal{R}'$ and $CC(\mathcal{S} \otimes \mathcal{S}')\leq CC(\mathcal{R}\otimes \mathcal{R}')$. The rank lower bound \cite[Lemma 1.28]{KN97} also implies that $
CC(\mathcal{R} \otimes \mathcal{R}) \ge \log(\rank(M_{\mathcal{R}}))+\log(\rank(M_{\mathcal{R}'}))$, where $M_{\mathcal{R}}$ and $M_{\mathcal{R}'}$ are the communication matrices of $\mathcal{R}$ and $\mathcal{R}'$, respectively.

\subsubsection{Simplification lemma for communication matrices}\label{subsec:simplification-lem}
We next present a lemma that will be important in our analysis of the communication matrices of the hazard-derivatives. 

Let $Z$ be a finite set. Denote by $\mathbf{P}(Z)$ the power set of $Z$. Let $M\in \mathbf{P}(Z)^{k\times \ell}$ and $M'\in \mathbf{P}(Z)^{k'\times \ell}$ be matrices with entries in $\mathbf{P}(Z)$ such that $k\ge k'$. Additionally, let $X,Y$ be finite sets such that $\mathcal{X}, \mathcal{X}'\subseteq X$ are the row labels and $\mathcal{Y}, \mathcal{Y}'\subseteq Y$ are the column labels of $M$ and $M'$, respectively. Note that the row and column labels of $M$ and $M'$ are allowed to be different.
The matrices $M$ and $M'$ can be naturally associated with communication problems. 

\begin{definition}[A sub-row and a super-row]\label{def:sub-row-super-row}
Let $M$ and $M'$ be matrices as in the above setting. Let $i\in[k]$ and $i'\in[k']$. We say that a row $M'_{i'}$ is a \emph{sub-row} of $M_i$ if the following holds:
\[
\forall j\in [\ell], \quad M'_{i',j}\subseteq M_{i,j}.
\]
We say that $M_i$ is a \emph{super-row} in $M'$ if there exists some $i'\in[k']$ such that $M'_{i'}$ is a sub-row of $M_i$.
\end{definition}

\begin{lemma}[Communication matrix simplification]\label{lem:simplification-lemma}
Let $M$ and $M'$ be matrices as in the above setting. W.l.o.g, assume that the rows in $M$ and the rows in $M'$ are distinct. Moreover, assume:
\begin{enumerate}
    \item For every $i\in[k]$, $M_i$ is a super-row in $M'$. \label{item:1}
    \item For every $i'\in [k']$ there exists $i\in [k]$ such that $M'_{i'}=M_{i}$. \label{item:2}
\end{enumerate}
Then, 
\[
\monorect(M)=\monorect(M'), \quad CC(M)=CC(M').
\]
\end{lemma}
\begin{proof}
Let $P'$ be a protocol for $M'$. We construct a protocol $P$ for $M$ as follows: let $x\in \mathcal{X}$ be Alice's input and $y\in \mathcal{Y}$ be Bob's input. Say $x$ is the label for the $i$'th row of $M$, and $y$ is a label for the $j$'th column in $M$. From assumption \ref{item:1}, there exists $i'\in[k']$ such that $M'_{i'}$ is a sub-row of $M_{i}$. Say $x' \in \mathcal{X}'$ is the label for the row $M'_{i'}$. Note that Alice can translate $x$ to $x'$ without any knowledge of Bob's input. Bob, independently of Alice's input, translates $y$ to $y'\in \mathcal{Y}'$, the label of the $j$'th column in $M'$. Next, Alice and Bob run $P'(x',y')$ and output the same answer. From \Cref{def:sub-row-super-row} and the definition of a deterministic protocol, we get that $P'(x',y')\in M'_{i',j}\subseteq M_{i,j}$, hence $P$ returns a correct answer. Therefore, $\monorect(M')\ge \monorect(M)$ and $CC(M')\ge CC(M)$.

The claim follows, since by assumption \ref{item:2} we get $M'$ is a submatrix of $M$ and so $\monorect(M)\ge \monorect(M')$ and $CC(M)\ge CC(M')$ is immediate.
\end{proof}

\subsubsection{Karchmer-Wigderson games}\label{subsec:kw-games}
Karchmer and Wigderson \cite{KW90} showed an equivalence between the depth and size of a formula and the cost of an associated communication search problem, both in the general setting and the monotone setting.

\begin{definition}[KW game \cite{KW90}]\label{def:classic-kw-game}
Let $\booleanf{f}{n}$ be a Boolean function. The Karchmer-Wigderson game of $f$, denoted $\kwgame_{f}$, is the following communication problem: Alice receives $\booleanx{x}{n}$ such that $f(x)=1$, Bob receives $\booleanx{y}{n}$ such that $f(y)=0$ and their goal is to determine a coordinate $i\in[n]$ such that $x_{i} \neq y_{i}$.
\end{definition}

Note that for a Boolean function $\booleanf{f}{n}$, the communication problem can be defined as the relation $\kwgame_{f}\subseteq X\times Y\times Z$ where $X=f^{-1}(1), Y=f^{-1}(0)$ and $Z=[n]$. 

\begin{theorem}[Communication vs. formula \cite{KW90}]\label{thm:kw-classic}
Let $\booleanf{f}{n}$ be a Boolean function. Then,
\[
    \text{$\depth_{F}(f)=CC(\kwgame_{f})$ and $\size_{F}(f)=\monorect(\kwgame_{f})$} .
\]
\end{theorem}

Similarly, a monotone version of the game is defined, with the distinction that Alice and Bob must output a coordinate where Alice's input has a 1 and Bob's input has a 0. This version of the game captures the complexity of monotone computation:

\begin{definition}[Monotone KW game \cite{KW90}]\label{def:monotone-kw-game}
Let $\booleanf{f}{n}$ be a monotone Boolean function. The Monotone Karchmer-Wigderson game of $f$, denoted $\kwgame_{f}^{+}$, is the following communication problem: Alice receives $\booleanx{x}{n}$ such that $f(x)=1$, Bob receives $\booleanx{y}{n}$ such that $f(y)=0$ and their goal is to determine a coordinate $i\in[n]$ such that $1=x_{i} \neq y_{i}=0$.
\end{definition}

\begin{theorem}[Communication vs. monotone formula \cite{KW90}]\label{thm:kw-monotone}
Let $\booleanf{f}{n}$ be a monotone Boolean function. Then,
\[
    \text{$\depth_{F}^{+}(f)=CC(\kwgame_{f}^{+})$ and $\size_{F}^{+}(f)=\monorect(\kwgame_{f}^{+})$}.
\]
\end{theorem}

\section{KW game of the hazard-derivative}\label{sec:kw-games-of-the-derivative}
In this section, we prove new insights on the $\kwgame^{\u}$ communication matrices of the hazard-derivatives, while relying on the framework of hazard-free KW games restricted to the prime implicants and prime implicates (see \Cref{thm:kw-hazard-free-implicants}). Our main observation is that for every $\booleanf{f}{n}$ and $\booleanx{x}{n}$, the communication matrix of $\kwgame^+_{\dfunc{f}{x}}$ is a submatrix of $\matkwu{f}{}$, that captures the communication complexity of a (potentially larger) submatrix of $\matkwu{f}{}$.

We present a simple example which serves as an illustration throughout this section. Consider $\booleanf{\fprime}{4}$, defined such that for every $k\in\mathbb{N}, 0\le k\le 15$, $\fprime(\text{bin}(k))=1$ if and only if $k$ is a prime number. One can calculate the prime implicants and prime implicates of $\fprime$ and sketch its hazard-free KW game communication matrix, as demonstrated in \Cref{subfig:example-f-communication-matrix}.

\begin{figure}[t]
    \centering
        \begin{subfigure}[t]{0.5\textwidth}
        \centering
        \begin{minipage}{\textwidth}
        $$
        \begin{pNiceMatrix}[first-col,first-row,margin]
        & \u\u 00 & \u 1 \u 0 & 1 \u \u 0 & \u 00 \u & 111\u \\
        001\u & \Block[tikz={preaction={fill=color1!20}, pattern=north west lines, pattern color=white}]{2-3}{} 3 & 2 & 1 & \Block[tikz={fill=color2!50}]{1-1}{} 3 & \Block[tikz={fill=color3!50}]{1-1}{} 1,2 \\
        \u 101 & 4 & 4 & 4 & \Block[tikz={preaction={fill=color2!20}, pattern=north west lines, pattern color=white}]{2-1}{} 2 & \Block[tikz={preaction={fill=color3!20}, pattern=north west lines, pattern color=white}]{3-1}{} 3 \\
        \u 011 & \Block[tikz={fill=color1!50}]{3-3}{} 3,4 & 2,4 & 4 & 3 & 2 \\
        0 \u 11 & 3,4 & 4 & 1,4 & \Block[tikz={fill=color2!50}]{2-1}{} 3 & 1 \\
        01 \u 1 & 4 & 4 & 1,4 & 2 & \Block[tikz={fill=color3!50}]{1-1}{} 1 \\
        \end{pNiceMatrix}
        $$
        \end{minipage}
        \caption{The matrix of $\fprime$.}
        \label{subfig:example-f-communication-matrix}
        \end{subfigure}
        \vfill
        \begin{subfigure}[b]{0.32\textwidth}
        \centering
        \begin{minipage}{\textwidth}
        $$
        \begin{pNiceMatrix}[first-col,first-row,margin]
          & \u\u 00 & \u0\u0 & 0\u\u0 \\
          111\u & \Block[tikz={preaction={fill=color1!20}, pattern=north west lines, pattern color=white}]{2-3}{} 3 & 2 & 1 \\
          \u\u\u 1 & 4 & 4 & 4 \\
        \end{pNiceMatrix}
        $$
        \end{minipage}
        \caption{The matrix of $\dfunc{\fprime}{1100}$.}
        \label{subfig:example-f-1100-communication-matrix}
        \end{subfigure}
        \hfill
        \begin{subfigure}[b]{0.32\textwidth}
        \centering
        \begin{minipage}{\textwidth}
        $$
        \begin{pNiceMatrix}[first-col,first-row,margin]
          & \u 00 \u \\
          \u1\u\u & \Block[tikz={preaction={fill=color2!20}, pattern=north west lines, pattern color=white}]{2-1}{} 2 \\
          \u\u1\u & 3 \\
        \end{pNiceMatrix}
        $$
        \end{minipage}
        \caption{The matrix of $\dfunc{\fprime}{1001}$.}
        \label{subfig:example-f-1001-communication-matrix}
        \end{subfigure}
        \hfill
        \begin{subfigure}[b]{0.32\textwidth}
        \centering
        \begin{minipage}{\textwidth}
        $$
        \begin{pNiceMatrix}[first-col,first-row,margin]
          & 000 \u \\
          \u\u1\u & \Block[tikz={preaction={fill=color3!20}, pattern=north west lines, pattern color=white}]{3-1}{} 3 \\
          \u1\u\u & 2 \\
          1\u\u\u & 1 \\
        \end{pNiceMatrix}
        $$
        \end{minipage}
        \caption{The matrix of $\dfunc{\fprime}{1111}$.}
        \label{subfig:example-f-1111-communication-matrix}
        \end{subfigure}
    \caption{Sketch of $\kwgame^u$ communication matrices for $\fprime$.}
    \Description{A sketch of the $\kwgame^u$ communication matrices for $\fprime$ and three of its hazard-derivatives at inputs $1100$, $1001$, and $1111$. It is observed that each communication matrix of a hazard-derivative is a submatrix of the communication matrix of the original function. Furthermore, the rows of the submatrices are unique, while the rows in the communication matrix of $\fprime$ that are not included are super-rows (see \Cref{def:sub-row-super-row}) in the communication matrix of the corresponding hazard-derivative.}
    \label{fig:example-f}
\end{figure}

Let $\booleanf{f}{n}$ be a non-constant Boolean function and let $\booleanx{x}{n}$ be a fixed Boolean input such that $f(x)=b\in\{0,1\}$. We begin by proving the following connection between the prime implicants and prime implicates of $f$ and $\dfunc{f}{x}$:

\begin{restatable}[The prime implicates of the hazard-derivative]{lemma}{lemprimeimplicatesdf}\label{lem:prime-implicates-df}
\[
\implicates{\dfunc{f}{x}} = \{p_{b} + x : p_{b} \in \implic{f}{b}|_x\}.
\]
\end{restatable}

\begin{restatable}[The prime implicants of the hazard-derivative]{lemma}{lemprimeimplicantsdf}\label{lem:prime-implicants-df}
\[
    \implicants{\dfunc{f}{x}} \subseteq \{\replaceall{(p_{\neg b} + x)}{0}{\u} : \text{$p_{\neg b}\in \implic{f}{\neg b}$}\} \subseteq I_1^{(\dfunc{f}{x})}.
\]
\end{restatable}

To demonstrate the claims consider \Cref{subfig:example-f-1100-communication-matrix}. For the Boolean input $1100$, it holds that $\fprime(1100)=0$. It can be verified that $\implicates{\fprime}|_{1100}=\{\u\u00, \u1\u 0, 1\u\u 0\}$. By \Cref{lem:prime-implicates-df}, the prime implicates of $\dfunc{\fprime}{1100}$ are $\{\u\u00+1100, \u1\u 0+1100, 1\u\u 0+1100\}=\{\u\u00, \u0\u 0, 0\u\u 0\}$. This result makes sense: apart from $0010$, even numbers are not prime. To transform $1100$ to $0010$, all three of its most significant bits must be perturbed. Therefore, any perturbation affecting at most two of these three bits preserves the function’s value. Perturbing the least significant bit yields $1101$, the prime number $13$, so it must  remain unchanged. Following \Cref{lem:prime-implicants-df}, calculating $\replaceall{(p_{1}+1100)}{0}{\u}$ for every $p_1\in\implicants{\fprime}$ yields $\{111\u, \u 111, 1\u 11, 1\u\u1, \u\u\u 1\}$. This set does not consist solely of \emph{prime} implicants, as $\u 111, 1\u 11, 1\u\u1$ are already ``covered'' by $\u\u\u 1$. Furthermore, this is a subset of $I_1^{(\dfunc{\fprime}{1100})}$ as all but $111\u$ perturb the least significant bit of $1100$, and $R(110\u)$ contains $1101$. Similarly, $111\u$ perturbs (at least) the three most significant bits, and $R(\u\u\u0)$ contains $0010$, which is the prime number $2$.

Let us discuss the underlying intuition before providing the full proofs. The hazard-derivative is a monotone function and so it is enough to consider the maximal and the minimal resolutions with respect to $\leq$ (see \Cref{subsec:prime-implicants-monotone} for details).

Consider $q_0\in\implicates{\dfunc{f}{x}}$. The $1$s in $\hat{q_0}$ represent all the coordinates in $x$ that \emph{can} be perturbed while the output $\tilde{f}(x+\u\hat{q_0})$ still equals $f(x)=b$. Since $q_0$ is prime, it should not be possible to add more $\u$s to it, and consequently to $x+\u\hat{q_0}$. This is the same as saying that $x+\u\hat{q_0}=p_b \in \implic{f}{b}|_x$. Since $\u\hat{q_0}=q_0$, it follows that $q_0=x+p_b$. The other direction is more intuitive: for $p_b \in \implic{f}{b}|_x$, the $0$'s in $x + p_b$ encode exactly the stable coordinates in $p_b$. This is a minimal subset of coordinates that must \emph{not} be perturbed in $x$ if we wish to get the output $f(x)$.

Consider $q_1\in \implicants{\dfunc{f}{x}}$. The $1$s in $\check{q_1}$ represent all the coordinates in $x$ that \emph{must} be perturbed in order for $\tilde{f}(x+\u\check{q_1})$ to be equal $\u$. Hence, there must be a resolution $z\in R(x+\u\check{q_1})$ for which $f(z)\neq f(x)$. We prove in \Cref{lem:prime-implicants-df-aux} that for every $q_1\in \implicants{\dfunc{f}{x}}$ such $z$ is unique among $R(x+\u\check{q_1})$. Specifically, $z$ is the resolution that differs from $x$ in the the maximal number of coordinates possible, that is, in all coordinates for which $\check{q_1}$ equals $1$. Since $f(z)=\neg b$ there is some $p_{\neg b}\in \implic{f}{\neg b}$ such that $p_{\neg b}\precequ z$. We then show that $q_1=\replaceall{(p_{\neg b}+x)}{0}{\u}$, relying on the fact that $z$ is unique.

\begin{proof}[Proof of \Cref{lem:prime-implicates-df}]
$\subseteq$: Let $q_{0} \in \implicates{\dfunc{f}{x}}$. Let $\hat{q_0}$ be the maximal resolution of $q_0$ as in \Cref{def:max-and-min-res}. By \Cref{lem:derivative-monotone} and \Cref{cor:prime-monotone}, we have that $q_0\in \{0,\u\}^n$ and therefore, $q_0=\u \hat{q_0}$. 
Following \Cref{lem:prime-implicants-implicates-monotone}, $q_{0} \in \implicates{\dfunc{f}{x}}$ if and only if:
\begin{enumerate}
    \item $\df{f}(x;\hat{q_0}))=0$. \label{item:prime-implicates-df-1}
    \item For every coordinate $i\in[n]$ such that $(q_0)_i\neq \u$, it holds that $\df{f}(x;\replace{\hat{q_0}}{i}{1}))=1$. \label{item:prime-implicates-df-2}
\end{enumerate}

Assumption \ref{item:prime-implicates-df-1} implies that $\tilde{f}(x+\u\hat{q_0})=f(x)=b$, therefore, $\tilde{f}(x+q_0)=b$. By assumption \ref{item:prime-implicates-df-2}, for every $i\in [n]$ that is a stable coordinate in $q_0$, we get  $\tilde{f}(x+\u(\replace{\hat{q_0}}{i}{1}))=\tilde{f}(x+\replace{(\u\hat{q_0})}{i}{\u})=\tilde{f}(\replace{(x+\u\hat{q_0})}{i}{\u})=\tilde{f}(\replace{(x+q_0)}{i}{\u})=\u$.
As a coordinate $i$ is stable in $x+q_0$ if and only if it is stable in $q_0$, it follows that $(x+q_0) \in \implic{f}{b}$. 

Denote $p_b:=x+q_0$, therefore $p_b+x=(x+q_0)+x=q_0$ and $p_b \precequ x$ as required.

$\supseteq$: Let $p_{b} \in \implic{f}{b}|_x$. Denote $q:= p_b+x$. As $x\in R(p_b)$, we have for every $i\in [n]$:
\begin{equation}\label{eq:df-prime-implicate}
    (q)_{i}=\begin{cases}
        0, & (p_{b})_{i}\neq \u,\\
        \u, & (p_{b})_{i}=\u.
    \end{cases}
\end{equation}
We show that $q \in P_0^{\dfunc{f}{x}}$ by proving the conditions of \Cref{lem:prime-implicants-implicates-monotone} hold. First, following \eqref{eq:df-prime-implicate} and \Cref{def:max-and-min-res}, we have $x+\u\hat{q}=x+q=p_b$. Since $p_b\in P_b^{(f)}$, we have $\tilde{f}(x+\u\hat{q})=\tilde{f}(p_b)=f(x)=b$ and therefore $\df{f}(x;\hat{q})=0$. 

Next, let $i\in [n]$ be a stable coordinate in $q$. By \eqref{eq:df-prime-implicate}, it is also a stable coordinate in $p_b$. Since $x+\u(\replace{\hat{q}}{i}{1})=\replace{(x+\u\hat{q})}{i}{\u}=\replace{p_b}{i}{\u}$, and $p_b\in P_b^{(f)}$, we have $\tilde{f}(x+\u(\replace{\hat{q}}{i}{1}))=\tilde{f}(\replace{p_b}{i}{\u})=\u$ and therefore $\df{f}(x;\replace{\hat{q}}{i}{1})=1$.  
\end{proof}

Before giving a full proof for \Cref{lem:prime-implicants-df}, we prove the following auxiliary lemma:

\begin{lemma}\label{lem:prime-implicants-df-aux}
Let $\booleanf{f}{n}$ be a Boolean function, and let $\booleanx{x}{n}$ be a fixed Boolean input such that $f(x)=b\in\{0,1\}$. Then, $q_1\in \implicants{\dfunc{f}{x}}$ if and only if $q_1\in \{1,\u\}^n$ and the following $z\in  R(x+\u\check{q_1})$ is the only resolution for which $f(z)=\neg b$:
\[
    z\quad :=\quad \forall i\in [n], \quad (z)_{i}=\begin{cases}
        x_{i}, & (\check{q_1})_{i} = 0, \\
        \neg x_{i}, & (\check{q_1})_{i} = 1.
    \end{cases}
\]
\end{lemma}
\begin{proof}
$\Rightarrow$ By \Cref{lem:prime-implicants-implicates-monotone} it holds that: 
\begin{enumerate}
    \item $\tilde{f}(x+\u\check{q_1})=\u$. \label{item:prime-implicants-df-aux-1}
    \item For every $i\in [n]$ such that $(q_1)_i\neq \u$, $\tilde{f}(x+\u(\replace{\check{q_1}}{i}{0}))=f(x)=b$. \label{item:prime-implicants-df-aux-2}
\end{enumerate}
Assume towards contradiction there exists $z'\in R(x+\u\check{q_1})$ such that $z'\neq z$ and $f(z')=\neg b$. There exists $j\in [n]$ such that $z_j\neq (z')_j$. For every $i\in [n]$, when $(\check{q_1})_i=0$ we have that $z_i=(z')_{i}=x_i$. Hence, $(\check{q_1})_j=1$, $(z')_{j}=x_{j}$, $z_j=\neg x_j$ and $(q_1)_j=1$. Therefore, $z'\in R(x+\u(\replace{\check{q_1}}{j}{0}))$ and by assumption \ref{item:prime-implicants-df-aux-2}, $f(z')=b$ in contradiction. As $\tilde{f}(x+\u\check{q_1})=\u$, it must be the case that  $f(z)=\neg b$. \Cref{cor:prime-monotone} completes the proof. 

$\Leftarrow$ Since $x,z\in R(x+\u\check{q_1})$, it holds that $\tilde{f}(x+\u\check{q_1})=\u$, implying $\df{f}(x;\check{q_1})=1$. Let $i\in [n]$ such that $(q_1)_i\neq \u$, so $(q_1)_i=1$ and $(\check{q_1})_i=1$. Since $R(x+\u(\replace{\check{q_1}}{i}{0}))\subsetneq R(x+\u\check{q_1})$ and $z\not\in R(x+\u(\replace{\check{q_1}}{i}{0}))$ we get from our assumption that $\tilde{f}(x+\u(\replace{\check{q_1}}{i}{0}))=f(x)=b$. Hence, $\df{f}(x;\replace{\check{q_1}}{i}{0})=0$. Following \Cref{lem:prime-implicants-implicates-monotone}, we get $q_1\in \implicants{\dfunc{f}{x}}$.
\end{proof}

\begin{proof}[Proof of \Cref{lem:prime-implicants-df}]
Let $q_{1} \in \implicants{\dfunc{f}{x}}$. Following \Cref{lem:prime-implicants-df-aux}, we have $q_1\in\{1,\u\}^n$ and $z\in R(x+\u\check{q_1})$ such that $f(z)=\neg b$ as described. Let $p_{\neg b} \in \implic{f}{\neg b}$ such that $p_{\neg b} \precequ z$.
By \Cref{lem:prime-implicants-df-aux}, for $i\in [n]$ we have that:
\[
    (p_{\neg b}+x)_{i}=\begin{cases}
        \u, & (p_{\neg b})_{i}=\u,\\
        z_{i} + x_{i}=x_i+x_i=0, & (p_{\neg b})_{i} \neq \u, (\check{q_1})_{i}=0, \\
        z_{i} + x_{i}=\neg x_{i} + x_{i}=1, & (p_{\neg b})_{i} \neq \u, (\check{q_1})_{i}=1.
    \end{cases}
\]
This implies that if $(p_{\neg b}+x)_{i}=1$ then $(\check{q_1})_{i}=1$. We next show that this is an if and only if statement: Assume towards contradiction that there exists $i\in[n]$ such that $(p_{\neg b}+x)_{i}=\u$ and $(\check{q_1})_{i}=1$. Then $z'=\replace{z}{i}{x_{i}}$ satisfies $p_{\neg b} \precequ z'$ so $f(z')=\neg b$. As $(\check{q_1})_i=1$ we have $z_i=\neg x_i$ and so $(z')_i\neq z_i$. Since $z'\in R(x+\u\check{q_1})$ and $z'\neq z$ we get a contradiction to \Cref{lem:prime-implicants-df-aux}.

Consequently, $(p_{\neg b}+x)_{i}=1 \iff (\check{q_1})_{i}=1 \iff (q_{1})_{i}=1$. Since $q_{1}\in\{1,\u\}^n$, we get $\replaceall{(p_{\neg b} + x)}{0}{\u}=q_{1}$. This holds for every $q_{1} \in \implicants{\dfunc{f}{x}}$, so we get the first containment.

Let $p_{\neg b}\in \implic{f}{\neg b}$ and denote $p:=\replaceall{(p_{\neg b}+x)}{0}{\u}$. We prove that $p\in I_{1}^{(\dfunc{f}{x})}$. As $\dfunc{f}{x}$ is monotone, following \Cref{lem:implicants-implicates-monotone} it is enough to show that $\check{p}$ satisfies $\df{f}(x;\check{p})=1$. Therefore, we prove that $\tilde{f}(x+\u\check{p})=\u$. It holds that:
\[
\forall i\in [n], \quad (x+\u\check{p})_i=\begin{cases}
        \u, & (p_{\neg b})_{i}=\neg x_i, \\
        x_i, & \text{$(p_{\neg b})_{i}=x_i$ or $(p_{\neg b})_{i}=\u$}.
    \end{cases}
\]
Let $\booleanx{\tilde{z}}{n}$ such that:
\[
\forall i\in [n], \quad \tilde{z}_i:=\begin{cases}
        \neg x_i, & (p_{\neg b})_{i}=\neg x_i, \\
        x_i, & \text{$(p_{\neg b})_{i}=x_i$ or $(p_{\neg b})_{i}=\u$}.
    \end{cases}
\]
Clearly, $\tilde{z}\in R(p_{\neg b})$, therefore $f(x')=\neg b$. Moreover, $\tilde{z}, x\in R(x+\u\check{p})$ and so $\tilde{f}(x+\u\check{p})=\u$.
\end{proof}

\begin{remark}
Note that unlike \Cref{lem:prime-implicates-df}, in \Cref{lem:prime-implicants-df} we do not get a complete characterization of the prime implicants of the hazard-derivative. There may be $p_{\neg b}\in\implic{f}{\neg b}$ such that for $p:=\replaceall{(p_{\neg b}+x)}{0}{\u}$ there are multiple $z\in R(x+\u\check{p})$ such that $f(z)\neq f(x)$. In this case, we prove that $p$ is an implicant of $\dfunc{f}{x}$ (but not prime). See \Cref{ex:prime-implicants-df}.
\end{remark}

\begin{example}[In \Cref{lem:prime-implicants-df} equality does not necessarily hold]\label{ex:prime-implicants-df} 
We present a simple example where $\implicants{\dfunc{f}{x}}\neq \{\replaceall{(p_{\neg b} + x)}{0}{\u} : \text{$p_{\neg b}\in \implic{f}{\neg b}$}\}$.  
Consider $\booleanf{f}{3}$ such that $f^{-1}(1)=\{000,001,011,111\}$. It is not hard to see that $\implicants{f}=\{\u 11,0 \u 1,00 \u\}$. When setting $x:=110$, we get $b:=f(x)=0$ and:
\[
\{\replaceall{(p_{\neg b} + x)}{0}{\u} : \text{$p_{\neg b}\in \implic{f}{\neg b}$}\} =\{\u\u 1,1 \u 1,11 \u\}.
\]
Clearly, $1 \u 1$ is not a prime implicant of $\dfunc{f}{x}$, since $\u\u 1$ is also an implicant. Moreover, both $001,011\in R(\u\u 1)$. Therefore, $\{\u\u 1,1 \u 1,11 \u\}\neq \implicants{\dfunc{f}{x}}$.
\end{example}

We showed how to infer $\implicants{\dfunc{f}{x}}$ and $\implicates{\dfunc{f}{x}}$ from $\implicants{f}$ and $\implicates{f}$. In fact, every $q_1\in \implicants{\dfunc{f}{x}}$ is produced by one (or more) $p_{\neg b}\in \implic{f}{\neg b}$, and every $q_0\in \implicates{\dfunc{f}{x}}$ is produced by exactly one $p_b\in P_b^{(f)}$. 
This translation preserves the set of stable and different coordinates. As a result, these sets are the equal for $p_{\neg b},p_b$ and $q_1,q_0$.

\begin{restatable}{lemma}{lemmdfsubmatrixmf}\label{lem:mdf-submatrix-mf} 
Let $\booleanx{x}{n}$ such that $f(x)=b\in\{0,1\}$. Let $q_1\in\implicants{\dfunc{f}{x}}$ and $q_0\in\implicates{\dfunc{f}{x}}$. According to \Cref{lem:prime-implicants-df} and \Cref{lem:prime-implicates-df} there exist $p_b\in\implic{f}{b}|_x$ and $p_{\neg b} \in \implic{f}{\neg b}$ such that $q_1=\replaceall{(p_{\neg b}+x)}{0}{\u}$ and $q_0=p_b+x$.
Then, for every $i\in [n]$ the following holds:
\[
(q_1+q_0)_i=1 \iff (p_{\neg b}+p_b)_i=1.
\]
\end{restatable}
\begin{proof}
$(q_1+q_0)_i=1$ iff\footnote{By \Cref{cor:prime-monotone}, $\dfunc{f}{x}$ is a monotone function.} $(q_{1})_{i}=1, (q_{0})_{i}=0$ iff $(p_{\neg b})_{i}= \neg x_{i}, (p_{b})_{i}=x_{i}$ iff\footnote{By our assumptions we have that $p_{b} \precequ x$.} $(p_{\neg b})_{i}, (p_{b})_{i} \neq \u$ and $(p_{b})_{i} \neq (p_{\neg b})_{i}$ iff $(p_{\neg b}+p_b)_i=1$.
\end{proof}

\begin{figure}[t]
        \centering
        \begin{tikzpicture}[overlay, decoration={brace,amplitude=10pt,mirror}]
            \draw[decorate] (0.1,-2.15) -- node[right=10pt, font=\large] {$M_x$} (0.1,-0.05);
        \end{tikzpicture}
        \begin{tikzpicture}[overlay]
            \node[font=\small] at (-2.30,-1.1) {$\implic{f}{\neg b}$};
            \node[font=\tiny] at (-0.9,0.2) {$P_b^{(f)}|_x$};
            \node[font=\tiny] at (0.8,0.2) {$P_b^{(f)}\setminus P_b^{(f)}|_x$};
        \end{tikzpicture}
        \begin{minipage}{\textwidth}
        $$
        \begin{pNiceArray}[margin]{ccccc|ccccc}
          \Block[tikz={preaction={fill=color1!20}, pattern=north west lines, pattern color=white}]{2-5}<\small>{\matkwu{\dfunc{f}{x}}{}} & & & & & \\
          & & & & & & & & & \\
          \Block[tikz={fill=color1!50}]{3-5}{} & & & & & & \\
          & & & & & & & & & \\
          & & & & & & & & & \\
        \end{pNiceArray}
        $$
        \end{minipage}
\caption{$\matkwu{f}{}$, $M_x$ and $\matkwu{\dfunc{f}{x}}{}$.}
\label{subfig:matrics-demonstrations-a}
\Description{A sketch of $\matkwu{f}{}$ and the placement of $M_x$ and $\matkwu{\dfunc{f}{x}}{}$: $\matkwu{\dfunc{f}{x}}{}$ is a submatrix of $M_x$, which in turn is a submatrix of $\matkwu{f}{}$.}
\end{figure}
A concrete example of \Cref{lem:mdf-submatrix-mf} is shown in Figures~\ref{subfig:example-f-communication-matrix} and~\ref{subfig:example-f-1100-communication-matrix}, where the same colors as in \Cref{subfig:matrics-demonstrations-a} are used for $M_{1100}$ and $\matkwu{\dfunc{\fprime}{1100}}{}$. It can be verified that $\u1\u0 \in \implicates{\fprime}$ yields $\u0\u0\in \implicates{\dfunc{\fprime}{1100}}$ and $001\u \in \implicants{\fprime}$ yields $111\u\in \implicants{\dfunc{\fprime}{1100}}$ with the corresponding entries both equal to $\{2\}$.

Let $\matkwu{f}{}$ and $\matkwu{\dfunc{f}{x}}{}$ denote the communication matrix of the hazard-free KW game of $f$ and $\dfunc{f}{x}$, respectively. Let $M_x$ be the submatrix of $\matkwu{f}{}$ with labels $\implic{f}{\neg b} \times P_b^{(f)}|_x$. We prove the following, as illustrated in \Cref{subfig:matrics-demonstrations-a}:
\begin{restatable}{proposition}{propmdfsubmatrixmf}\label{prop:mdf-submatrix-mf}
Let $\booleanf{f}{n}$ be a non-constant Boolean function. For every $\booleanx{x}{n}$, $\matkwu{\dfunc{f}{x}}{}$ is a submatrix of $\matkwu{f}{}$.
\end{restatable}
\begin{proof}
Fix $\booleanx{x}{n}$ such that $f(x)=b \in \{0,1\}$. W.l.o.g, assume $b=0$. Denote the rows labels of $\matkwu{\dfunc{f}{x}}{}$ by $\implicants{\dfunc{f}{x}}=\{q_{1}^{(1)}, \dots, q_{1}^{(k)}\}$ and the column labels by $\implicates{\dfunc{f}{x}}=\{q_{0}^{(1)}, \dots, q_{0}^{(\ell)}\}$. Similarly, denote the row labels of $\matkwu{f}{}$ by $\implic{f}{\neg b}=\{p_{\neg b}^{(1)}, \dots, p_{\neg b}^{(K)}\}$ and the column labels by $\implic{f}{b}=\{p_{b}^{(1)}, \dots, p_{b}^{(L)}\}$.

To prove the claim, we show that for every $i \in [k]$ there exists $\pi(i)\in [K]$ and for every $j\in [\ell]$ there exists $\sigma(j) \in [L]$ so that $(\matkwu{\dfunc{f}{x}}{})_{i,j}=(\matkwu{f}{})_{\pi(i), \sigma(j)}$. In other words, if we only consider the set $I=\{\pi(i):i\in [k]\}$ and the set $J=\{\sigma(j):j\in [\ell]\}$ then we get that, up to reordering,  $\matkwu{\dfunc{f}{x}}{}={\matkwu{f}{}}|_{I\times J}$.

Let $i \in [k]$. \Cref{lem:prime-implicants-df} implies that there exists $\pi(i)\in [K]$ such that $q_{1}^{(i)}=\replaceall{(p_{\neg b}^{(\pi(i))}+x)}{0}{\u}$. Similarly, let $j\in [\ell]$, then by \Cref{lem:prime-implicates-df}, there exists $\sigma(j)\in [L]$ such that $q_{0}^{(j)}=p_{b}^{(\sigma(j))}+x$. By \Cref{lem:mdf-submatrix-mf}, we have
$(\matkwu{\dfunc{f}{x}}{})_{i,j}=(\matkwu{f}{})_{\pi(i), \sigma(j)}$, which implies the claim.

Note that $\pi$ must be injective, since if $q_1^{(i)}\neq q_1^{(i')}$ then $p_{\neg b}^{(\pi(i))}\neq p_{\neg b}^{(\pi(i'))}$, otherwise, $q_1^{(i)}$ and $q_1^{(i')}$ would have been equal. Similarly, $\sigma$ must be injective.
\end{proof}

\begin{remark}\label{rem:mdf-submatrix-mf}
Observe that in the proof of \Cref{prop:mdf-submatrix-mf}, for every $j\in[\ell]$ it holds that $p_{b}^{(\sigma(j))}\in \implic{f}{b}|_x$. Thus, $\matkwu{\dfunc{f}{x}}{}$ is actually a submatrix of $M_x$.
\end{remark}

\Cref{prop:mdf-submatrix-mf} can be viewed as an alternative proof for \Cref{thm:derivative-lower-bound} for formulas:

\begin{corollary}\label{cor:derivative-lower-bound}
Let $\booleanf{f}{n}$ be a non-constant Boolean function and let $\booleanx{x}{n}$ be a fixed Boolean input. It holds that:
\[
\size_{F}^{+}(\dfunc{f}{x})\le \size_{F}^{\u}(f).
\]
\end{corollary}
\begin{proof}
Fix $\booleanx{x}{n}$. Say $P$ is a protocol that solves $\kwgame_{f}^{\u}$. 
We get that $P$ induces a partition of $\matkwu{f}{}$ to $\kwgame_{f}^{\u}$--monochromatic rectangles. We proved in \Cref{prop:mdf-submatrix-mf} that $\matkwu{\dfunc{f}{x}}{}$ is a submatrix of $\matkwu{f}{}$, and so the partition that $P$ induces on $\matkwu{f}{}$ is also a valid\footnote{A partition that can be induced by a protocol.} partition of $\matkwu{\dfunc{f}{x}}{}$. Note that $\dfunc{f}{x}$ is a monotone Boolean function according to \Cref{lem:derivative-monotone}. The claim follows since:
\begin{align}\label{eq:relating-size}
  \size_{F}^{\u}(f) & \stackrel{\Cref{thm:kw-hazard-free}}{=}\monorect(\kwgame_{f}^{\u}) \geq \monorect(\kwgame_{\dfunc{f}{x}}^{\u}) \\
   & \stackrel{\Cref{thm:kw-hazard-free-monotone}}{=} \monorect(\kwgame_{\dfunc{f}{x}}^{+})\stackrel{\Cref{thm:kw-monotone}}{=}\size_{F}^{+}(\dfunc{f}{x}).\qedhere
\end{align}
\end{proof}

Next, we make the simple but important observation. Certain rows, which we refer to as ``super-rows'', can be removed from a communication matrix without effecting its complexity. A \emph{super-row} is a row $r$ in the matrix such that there exists another row $r'$, called a \emph{sub-row}, whose every entry is a subset of the matching entry in $r$ (see \Cref{subsec:simplification-lem} for details). 

It is not too difficult to see that the rows in $M_x$ that are associated with $p_{\neg b}\in \implic{f}{\neg b}$ such that $\replaceall{(p_{\neg b}+x)}{0}{\u}\in I_1^{(\dfunc{f}{x})}\setminus \implicants{\dfunc{f}{x}}$ are super-rows in $M_x$. Implicants that are not prime, by definition, contain ``extra'' stable bits that can be replaced with $\u$ while preserving the output of the function. Therefore, these rows contain possibly more valid answers to the $\kwgame^\u$ on $\dfunc{f}{x}$ than the ones associated with $\implic{\dfunc{f}{x}}{1}$. By \Cref{lem:mdf-submatrix-mf}, we conclude that these rows are equal to those of $M_x$, and thus we are done. Figures~\ref{subfig:example-f-communication-matrix} and~\ref{subfig:example-f-1100-communication-matrix} demonstrate this claim, as the super-rows in~\ref{subfig:example-f-communication-matrix} correspond to $\{\u 011, 0\u11,01\u1\}$, which yielded $\{\u 111, 1\u 11, 1\u\u1\}$ in our previous calculation.

Following our observation, the super-rows can be removed without changing the complexity of $M_x$. Most importantly, removing them results in $\matkwu{\dfunc{f}{x}}{}$. Hence, $M_x$ and $\matkwu{\dfunc{f}{x}}{}$ are equivalent in terms of communication complexity. This highlights the importance of the hazard-derivative, which captures the communication complexity of $M_x$, which is a potentially larger submatrix of $\matkwu{}{}$.

\begin{restatable}{proposition}{propsubmatrixequivderivative}\label{prop:submatrix-of-derive-p-equiv-derivative}
Let $\booleanf{f}{n}$ and $\booleanx{x}{n}$ such that $f(x)=b\in\{0,1\}$. We denote by $M_x$ the submatrix of $\matkwu{f}{}$ with labels $\implic{f}{\neg b}\times \implic{f}{b}|_x$. The following holds:
\[
\monorect(M_x)=\monorect(\matkwu{\dfunc{f}{x}}{}), \quad
CC(M_x)=CC(\matkwu{\dfunc{f}{x}}{}).
\]
\end{restatable}
\begin{proof}
W.l.o.g, assume $b=0$. Following \Cref{lem:prime-implicates-df}, we denote the column labels of $M_x$ by $\implic{f}{b}|_{x}:=\{p_b^{(1)},\dots,p_b^{(\ell)}\}$ and the column labels of $\matkwu{\dfunc{f}{x}}{}$ by $\implicates{\dfunc{f}{x}}:=\{p_b^{(1)}+x,\dots,p_b^{(\ell)}+x\}$.
We prove that the assumptions of \Cref{lem:simplification-lemma} hold for $M:=M_x$ ,$M':=\matkwu{\dfunc{f}{x}}{}$, $k:=|\implic{f}{\neg b}|, k':=|\implicants{\dfunc{f}{x}}|, \ell:=|\implic{f}{b}|_x|$ and the $\kwgame^\u_f$ relation.
By \Cref{rem:mdf-submatrix-mf}, we get that $\matkwu{\dfunc{f}{x}}{}$ is a submatrix of $M_x$, so assumption \ref{item:2} of \Cref{lem:simplification-lemma} holds. 

We now prove that for every $i\in [k]$, $(M_x)_i$ is a super-row in $\matkwu{\dfunc{f}{x}}{}$.
Let $i\in[k]$ and let $p_{\neg b} \in 
\implic{f}{\neg b}$ be the associated row label in $M_x$. Following \Cref{lem:prime-implicants-df} we have $\replaceall{(p_{\neg b}+x)}{0}{\u}\in I_1^{(\dfunc{f}{x})}$, so there exists $q_1\in\implicants{\dfunc{f}{x}}$ such that $q_1\precequ \replaceall{(p_{\neg b}+x)}{0}{\u}$. We prove that the row associated with $q_1$ in $\matkwu{\dfunc{f}{x}}{}$ is a sub-row of $(M_x)_i$. By the first inclusion in \Cref{lem:prime-implicants-df}, there exists $\widetilde{p_{\neg b}}\in\implicants{f}$ such that $q_1=\replaceall{(\widetilde{p_{\neg b}}+x)}{0}{\u}$. We denote by $\tilde{i}\in [k]$ the row index of $\widetilde{p_{\neg b}}$ in $M_x$. 
According to \Cref{lem:mdf-submatrix-mf}, the row associated with $\widetilde{p_{\neg b}}$ in $M_x$ and the row associated with $q_1$ in $\matkwu{\dfunc{f}{x}}{}$ are equal, so it is enough to show that $(M_x)_{\tilde{i}}$ is a sub-row of $(M_x)_i$. Let $j\in [\ell]$ and $p_b^{(j)}\in \implic{f}{b}|_{x}$ be the associated column label in $M_x$. Recall that $p_b^{(j)}\precequ x$. Let a coordinate $d\in[n]$. Since $\replaceall{(\widetilde{p_{\neg b}}+x)}{0}{\u} \precequ\replaceall{(p_{\neg b}+x)}{0}{\u}$, we have $(\widetilde{p_{\neg b}})_d=\neg x_d \Rightarrow (p_{\neg b})_d=\neg x_d$, and so:
\begin{align*}
&d\in (M_x)_{\tilde{i},j}\Rightarrow(\widetilde{p_{\neg b}}+p_b^{(j)})_d=1\Rightarrow (\widetilde{p_{\neg b}})_d=\neg x_d, (p_b^{(j)})_d=x_d\\ 
&\Rightarrow(p_{\neg b})_d=\neg x_d, (p_b^{(j)})_d=x_d \Rightarrow (p_{\neg b}+p_b^{(j)})_d=1\Rightarrow d\in (M_x)_{i,j}.
\end{align*}
As for every $j\in [\ell]$, $(M_x)_{\tilde{i},j}\subseteq (M_x)_{i,j}$, $(M_x)_{i'}$ is a sub-row of $(M_x)_i$, therefore, assumption \ref{item:1} of \Cref{lem:simplification-lemma} holds.
\end{proof}

\section{Tightness of the hazard-derivative method}\label{sec:tightness-of-hdm}
In this section, we extend the hazard-derivative lower bound method of \cite{IKL+19} when applied to hazard-free formulas. We show that the hazard-derivative method yields \emph{exact} bounds if and only if the function is unate, thus proving \Cref{thm:kw-f-not-breaks-if}. Hence, whether a function is \emph{unate} or not, can serve as a criterion to achieving a monotone gap of 1, which in particular answers \Cref{que:monotone-gap-criterion}.

We prove \Cref{prop:kw-f-not-break-if} by applying \Cref{prop:submatrix-of-derive-p-equiv-derivative}, which immediately shows that the hazard-derivative lower bound method is tight for unate functions.

\begin{proposition}[Unate functions have monotone gap of 1]\label{prop:kw-f-not-break-if}
Let $\booleanf{f}{n}$ be a non-constant unate Boolean function. Then, there exists $\booleanx{x}{n}$ such that:
\[
\size_{F}^{+}(\dfunc{f}{x}) = \size_{F}^{\u}(f).
\]
\end{proposition}
\begin{proof}
A unate function is positive or negative unate in each of it's variables. By \Cref{lem:f-unate-implicates}, it is not possible for a unate function to have two prime implicates with \emph{stable and different} coordinates. Hence, there exists $x'\in \{0,1\}^n$ that satisfies $\implicates{f}|_{x'}=\implicates{f}$. In particular, 
$\implicants{f}\times \implicates{f}|_{x'}=\implicants{f}\times \implicates{f}$. Following \Cref{prop:submatrix-of-derive-p-equiv-derivative}, we get:
\[
CC(\matkwu{f}{})=CC(\matkwu{\dfunc{f}{x'}}{}),
\]
\[
\monorect(\matkwu{f}{})=\monorect(\matkwu{\dfunc{f}{x'}}{}).
\]
We conclude, as in Equation~\eqref{eq:relating-size}, that $\size_F^\u(f)=\size_F^+(\dfunc{f}{x'})$.
\end{proof}

\begin{remark}
An alternative proof of \Cref{prop:kw-f-not-break-if} is to generalize \cite[Corollary 4.5]{IKL+19}. A function $f$ is unate if and only if there exists some $\booleanx{c}{n}$ such that $f(x+c)$ is monotone. Therefore, it can be shown by \cite[Lemma 4.4]{IKL+19}, that $f(x+c)=\df{f}(c;x)$.
\end{remark}

The proof of the other direction requires more work. Consider a non-unate function $f$ and a coordinate $d\in [n]$ that demonstrates that $f$ is not unate. As before, let $\booleanx{x}{n}$ be a Boolean input such that $f(x)=b\in\{0,1\}$. 
Let $s\in\{0,\u,1\}$, and denote by $P_1^{d=s}$ ($P_0^{d=s}$) the subset of $P_1^{(f)}$ ($P_0^{(f)}$) for which the $d$'th coordinate equals $s$. By examining the prime implicants and the prime implicates of $f$ we observe that $P_{b}^{d=x_d},P_{b}^{d=\neg x_d},P_{\neg b}^{d=x_d},P_{\neg_b}^{d=\neg x_d}\neq \emptyset$. 
We actually prove the following stronger claim, as illustrated in \Cref{fig:non-unate-f}:

\begin{restatable}{proposition}{proprecsnonunate}\label{prop:recs-non-unate}
Let $\booleanf{f}{n}$ be a non-unate Boolean function. Then, there exists a coordinate $d\in[n]$, $p_1,q_1\in \implicants{f}$ and $p_0,q_0\in \implicates{f}$ such that:
\begin{enumerate}
    \item $(p_1)_d=0,(p_0)_d=1$ and $d$ is the only different and stable coordinate between $p_1$ and $p_0$.
    \item $(q_1)_d=1,(q_0)_d=0$ and $d$ is the only different and stable coordinate between $q_1$ and $q_0$.
\end{enumerate} 
\end{restatable}
\begin{proof}
As $f$ is non-unate, there exists $d\in [n]$ such that $f$ is not positive or negative unate in $x_{d}$. Hence, there exist $\booleanx{s,s'}{n}$ such that:
\[
    f(\replace{s}{d}{0}) > f(\replace{s}{d}{1}), f(\replace{s'}{d}{1}) > f(\replace{s'}{d}{0}).
\]
Therefore, there exist $\implicant, q_1 \in \implicants{f}$ such that $\implicant \precequ \replace{s}{d}{0}$ and $q_1 \precequ \replace{s'}{d}{1}$.
By \Cref{prop:implicant-and-neg-implicate}, there exist $p_0,q_0\in \implicates{f}$ such that only the $d$'th coordinate in $p_1,p_0$ and in $q_1,q_0$ is stable and different. As $0=(p_1)_d\neq (q_1)_d=1$, we conclude that $1=(p_0)_d\neq (q_0)_d=0$ and the claim follows.
\end{proof}


\begin{figure}[t]
    \centering
\vspace{1cm}
\begin{tikzpicture}[overlay, decoration={brace,amplitude=10pt,mirror}]
    \draw[decorate] (-2.9,-0.5) -- node[left=10pt] {$P_1^{d=1}$} (-2.9,-1.5);
    \draw[decorate] (-2.9,-1.55) -- node[left=10pt] {$P_1^{d=0}$} (-2.9,-2.55);
    \draw[decorate] (-2.9,-2.6) -- node[left=10pt] {$P_1^{d=\u}$}(-2.9,-3.25);
    
    \draw[decorate] (-0.25,-0.2) -- node[above=10pt] {$P_0^{d=1}$} (-2.15,-0.2);
    \draw[decorate] (1.85,-0.2) -- node[above=10pt] {$P_0^{d=0}$} (-0.2,-0.2);
    \draw[decorate] (2.55,-0.2) -- node[above=10pt] {$P_0^{d=\u}$}(1.9,-0.2);
\end{tikzpicture}
\begin{minipage}{\textwidth}
        $$
        \begin{pNiceMatrix}[first-col,first-row,margin]
        & p_0 & \dots & q_0 & \dots & \dots \\
        q_1 & & & \Block[tikz={fill=color4!40, draw=black}]{2-2}{} d & d,\dots & \\
        \vdots & & & d,\dots & d,\dots & \\
        p_1 & \Block[tikz={fill=color4!40, draw=black}]{2-2}{} d & d,\dots & & & \\
        \vdots & d,\dots& d,\dots & & & \\
        \vdots & & & & & \\
        \end{pNiceMatrix}
        $$
\end{minipage}
\caption{Sketch of $\matkwu{f}{}$ for a non-unate $f$. Empty entries do not contain the $d$'th coordinate. We use the notation $'d,\ldots'$ when there may be stable coordinates other than $d$ on which the row labels and column labels differ.}
\Description{The communication matrix of the hazard-free KW game for a non-unate Boolean function imposes a special structure, as per \Cref{prop:recs-non-unate}: in any partition of $\matkwu{f}{}$ to monochromatic rectangles, there must exist two monochromatic rectangles that are colored with $d$.}
\label{fig:non-unate-f}
\end{figure}

\begin{figure}[t]
    \centering
    \begin{tikzpicture}[overlay]
        \node[font=\tiny] at (-2.35,-0.4) {$P_{\neg b}^{d=\neg x_d}$};
        \node[font=\tiny] at (-2.35,-1.2) {$P_{\neg b}^{d=x_d}$};
        \node[font=\tiny] at (-2.35,-1.9) {$P_{\neg b}^{d=\u}$};
         \node[font=\tiny] at (-0.9,0.2) {$P_b^{(f)}|_x$};
        \node[font=\tiny] at (0.8,0.2) {$P_b^{(f)}\setminus P_b^{(f)}|_x$};
    \end{tikzpicture}
    \begin{minipage}{\textwidth}
    $$
    \begin{pNiceArray}[margin]{ccccc|ccccc}
      \Block[tikz={fill=color1!50}]{5-5}<\large>{M_x} & & & & & \\
      & & & & & & & & & \\
      & & & & & & & & & \\
      & & & & & & & & & \\
      & & & & & & & & & \\
    \end{pNiceArray}
    $$
    \end{minipage}
    \begin{tikzpicture}[overlay]
        \draw[fill=color4!40, draw=black] (0.7,1.425) rectangle node[font=\small]{$\rec$} (1.4,0.725);

        \draw[dashed] (-1.8,1.5) -- (1.9,1.5);
        \draw[dashed] (-1.8,0.65) -- (1.9,0.65);
    \end{tikzpicture}
\caption{$\matkwu{f}{}$ for a non-unate $f$.}
\label{subfig:matrics-demonstrations-b}
\Description{A depiction the communication matrix for a non-unate function: a $d$-uniform monochromatic rectangle $\rec$ is a subset of the entries in $P_{\neg b}^{d = x_d} \times (\implic{f}{b} \setminus \implic{f}{b}|_x)$.}
\end{figure}

By \Cref{prop:recs-non-unate}, we conclude that for every protocol solving $\kwgame_f^\u$, there exists a monochromatic rectangle $\rec \subseteq P_{\neg b}^{d=x_d}\times P_{b}^{d=\neg x_d}$ in the partition induced by the protocol that must be colored with $d$ and is disjoint from $M_x$, as illustrated in \Cref{subfig:matrics-demonstrations-b}. Removing the leaf associated with $\rec$ from the protocol tree results in a protocol for $M_x$ that has one less leaf. This implies that there must be some gap in the number of leaves in the optimal protocols, thus proving the following proposition:


\begin{restatable}
{proposition}{propkwfbreaksif}\label{prop:kw-f-breaks-if}
Let $\booleanf{f}{n}$ be a non-unate Boolean function. Then for every $\booleanx{x}{n}$, it holds that:
\[
    \size_{F}^{+}(\dfunc{f}{x}) < \size_{F}^{\u}(f).
\]
\end{restatable}
\begin{proof}
Let $\booleanx{x}{n}$ and denote $f(x)=b\in\{0,1\}$. 
W.l.o.g, assume $b=0$. Let $d\in [n]$, $p_{\neg b}\in P_{\neg b}^{d=x_d}$ and $p_{b}\in P_{ b}^{d=\neg x_d}$ as in \Cref{prop:recs-non-unate}.
According to \Cref{prop:mdf-submatrix-mf} and \Cref{rem:mdf-submatrix-mf}, $\matkwu{\dfunc{f}{x}}{}$ is a submatrix of $M_x$, the submatrix of $\matkwu{f}{}$ with labels $\implic{f}{\neg b}\times \implic{f}{b}|_x \subseteq \implic{f}{\neg b} \times (P_{b}^{d=x_d}\cup P_{b}^{d=\u})$. Let $P$ be a protocol for $\matkwu{f}{}$. As the $d$'th coordinate in $p_{\neg b},p_b$ is the only coordinate that is both stable and different, there exists a \emph{$d$-uniform} rectangle $\rec$ as in \Cref{def:d-colored-combinatorial-rec}, which contains $(p_{\neg b},p_b)$, in the partition induced by $P$. Moreover, $\rec$ has to be a subset of $P_{\neg b}^{d=x_b} \times P_b^{d=\neg x_d}$, since all entries of $\rec$ has to include the coordinate $d$. 

Next, consider the following protocol $P'$ for $\matkwu{\dfunc{f}{x}}{}$: Alice and Bob follow $P$, but since $P'$ operates only on $\implic{f}{\neg b} \times \implic{f}{b}|x$, we may get empty leaves which can be removed. We must have removed the leaf associated with $\rec$, so:
\[
\monorect(P')<\monorect(P).
\]
As before, the claim follows from \Cref{thm:kw-hazard-free}, \Cref{thm:kw-hazard-free-monotone} and  \Cref{thm:kw-monotone}. 
\end{proof}

The proof of \Cref{thm:kw-f-not-breaks-if} follows immediately from \Cref{prop:kw-f-not-break-if} and \Cref{prop:kw-f-breaks-if}.

\subsection{The monotone gap of Andreev's function}\label{subsec:andreev}
In this section, we prove that Andreev's function gives monotone gap of $\Omega\left(\frac{n^2}{\log n}\right)$.

\begin{definition}[Andreev's function \cite{And87}]\label{def:andreev-function}
Andreev's function, with parameters $k, m \in \mathbb{N}$, is a Boolean function over $n=2^k+km$ variables:
\[
\text{ANDREEV}_{k,m} : \{0,1\}^{2^k} \times \{0,1\}^{km} \rightarrow \{0,1\}.
\]
Andreev's function receives as inputs $\booleanx{f}{2^k}$ and $\booleanx{X}{km}$, a truth table of a Boolean function over $k$ bits and a Boolean $k\times m$ matrix, respectively: 
\[
\text{ANDREEV}_{k,m}(f, X) := (f \diamond \xor_m)(X),
\]
where $\booleanf{\xor_m}{m}$ is the parity function over $m$ bits:
\[
\xor_m(x):=x_1+\dots+x_m.
\]
\end{definition}

In the rest of this section, we shall use the following parameters $n:=2^k+km$, 
$k:=\log(n/2)$, and $m=\frac{n}{2\log(n/2)}$.\footnote{We can take $k$ and $n$ to be powers of $2$ so we ignore ceiling and floors for readability.} 

H{\aa}stad was the first to prove a near-cubic lower bound on the formula complexity of Andreev's function \cite{Has98} using the random restriction technique. Later, Tal \cite{Tal14} improved the analysis to obtain the currently best lower bound of $\Omega\left(\frac{n^3}{\log^2n\log\log n}\right)$. This lower bound is tight up to the $\log \log n$ factor.

\begin{theorem}[Andreev's function requires cubic formulas {\cite[Theorem 7.2]{Tal14}}]\label{thm:andreev-lower-bound}
\[
\Omega\left(\frac{n^3}{\log^2n\log\log n}\right)\le \size_F(\text{ANDREEV}_{k,m})\leq O\left(\frac{n^3}{\log^2n}\right).
\]
\end{theorem}

As $\size_F^\u(f)\ge \size_F(f)$, this lower bound also holds in the hazard-free setting. We next bound the complexity of the hazard-derivatives of Andreev's function.

\begin{fact}[Parity has linear hazard-derivatives]\label{fact:parity-der}
For every $\booleanx{x,y}{m}$, it holds that 
\[
\df{\xor_m}(x;y)=\bigvee_{i=1}^{m}y_i.
\]
\end{fact}

We denote $\df{\xor_m}(y):=\df{\xor_m}(x;y)$, since the hazard-derivative of $\xor_m$ w.r.t to $x$ does not depend on $x$.

\begin{claim}[Andreev's function has near linear derivatives]\label{claim:andreev-der}
For every $\booleanx{f}{2^k}$ and $\booleanx{X}{km}$, it holds that:
\[
\size_F^+(\df{\text{ANDREEV}_{k,m}(f,X;\cdot,\cdot)})=(k+1)2^k+km=O(n\log(n)).
\]
\end{claim}
\begin{proof}
By \Cref{def:mux}, we have:
\[
\text{ANDREEV}_{k,m}(f;X)=(f\diamond \xor_m)(X)=f(\xor_m[X])=\mux_k(f,\xor_m[X]).
\]

Let $\booleanx{f}{2^k}$ and $\booleanx{X}{km}$. We describe a monotone formula for the hazard-derivative of $\text{ANDREEV}_{k,m}$ w.r.t to $f$ and $X$, denoted $\dfunc{\text{ANDREEV}_{k,m}}{f,X}$. Let $\booleanx{t}{2^k}$ be a perturbation for $f$ and $\booleanx{Y}{km}$ be a perturbation for $X$. A direct consequence of \Cref{fact:parity-der} is $\xor_m(x+\u y)=\xor(x)+\u\df{\xor_m}(x;y)=\xor(x)+\u\df{\xor_m}(y)$, hence:
\[
\xor_m[X+\u Y]=\xor_m[X]+\u \df{\xor_m}[Y]\]

Therefore:
\begin{align*}
&\df{\text{ANDREEV}_{k,m}}(f,X;t,Y)=1 \iff \\
&\widetilde{\mux_k}(f+\u t,\xor_m[X+\u Y])=\u \iff \\
&\widetilde{\mux_k}(f+\u t,\xor_m[X]+\u \df{\xor_m}[Y])=\u \iff \\
&\df{\mux_k}(f,\xor_m[X];t,\df{\xor_m}[Y])=1.
\end{align*}

We can use a monotone formula for $\df{\xor_m}[Y]$, which is a simply $k$ disjoint $\text{OR}_m$ formulas, that feed into the monotone formula for $\df{\mux_k}(f, \xor_m[X];t,y)$, as described in \cite[Proposition 1]{IK23}.
\end{proof}

\begin{corollary}[Andreev's function has near quadratic monotone gap]
    \[
    \monogap(\text{ANDREEV}_{k,m})=\Omega\left(\frac{n^2}{\log n}\right).
    \]
\end{corollary}
\begin{proof}
Immediate result of \Cref{thm:andreev-lower-bound} and \Cref{claim:andreev-der}.
\end{proof}

\section{A universal formula upper bound for most Boolean functions}\label{sec:universal}
In this section, we study the hazard-free formula complexity of random Boolean functions. 

\begin{definition}\label{def:random-boolean-function}
We say that $\booleanf{f}{n}$ is a \emph{random function} if it is sampled uniformly at random among all Boolean functions on $n$ variables. Alternatively, we can think as generating $f$ by sampling, for every $\booleanx{x}{n}$, a value $f(x)\in \{0,1\}$ uniformly and independently.
\end{definition}

It is well known that the formula complexity of a random function is maximal. 
\begin{theorem}[{\cite{RS42,Sha49}}]\label{thm:almost-all-formulas-are-large}
For any $\epsilon>0$, almost all Boolean functions on $n$ variables require formula size of at least $\frac{(1-\epsilon)2^n}{\log n}$.
\end{theorem}
Since $\size_F^\u(f)\ge \size_F(f)$, an immediate consequence of \Cref{thm:almost-all-formulas-are-large} is that with high probability, a random Boolean function over $n$ variables requires hazard-free formulas of size at least $\Omega(\frac{2^n}{\log n})$.

In this section, we provide two proofs for \Cref{thm:universal-ub}, which greatly improves upon the best-known upper bound for arbitrary functions, currently $O(3^n)$ (see Sections ~\ref{subsec:random-function-dnf} and ~\ref{subsec:hazard-derivative-upper-bound}). The two proofs utilize the same observation: the probability that a random function is constant on any large Boolean subcube is small. Therefore, it tends to have prime implicants that consist mostly of stable bits, while its hazard-derivatives tend to be nearly constant and equal to $1$. The second proof consists of proving \Cref{thm:derivative-upper-bound}, which gives a weak converse to \Cref{thm:derivative-lower-bound} and offers insights of independent interest.

Another consequence of this observation is that the hazard-derivatives of random Boolean functions are significantly simpler than the functions themselves. This reveals a considerable gap between the hazard-free complexity of a random function and the monotone complexity of its hazard-derivatives. Consequently, random functions achieve a monotone gap that is exponential in $n$.

\subsection{Hazard-free DNF formula of a random Boolean function}\label{subsec:random-function-dnf}
We show that, with high probability, all the prime implicants of a random function consist mostly of stable bits. Hence, their number is $2^{(1+o(1))n}$ and a ``brute force'' DNF formula construction yields \Cref{thm:universal-ub}. 

\begin{lemma}[$\implicants{f}$ of a random $f$ tend to have more stable bits]\label{lem:weight-of-random-pi-f}
Let $\ternaryx{p_{1}}{n}$ be fixed such that $|p_1|_\u=k$. The probability that $p_{1}\in \implicants{f}$, for a random Boolean function $\booleanf{f}{n}$, is at most $(\frac{1}{2})^{2^k}$.
\end{lemma}
\begin{proof}
The probability that $p_1$ is a prime implicant of $f$ is at most the probability that it is an implicant of $f$. Furthermore, $p_1\in I_1^{(f)}$ if and only if for every $z\in R(p_1)$, it holds that $f(z)=1$. Since $|R(p_1)|=2^k$, we get that $p_1\in I_1^{(f)}$ with probability $\left(\frac{1}{2}\right)^{2^k}$.
\end{proof}

\begin{lemma}[Hazard-free DNF formula for a random function]\label{lem:dnf-for-random-f}
Let $\booleanf{f}{n}$ be a random Boolean function. Then, with high probability, $f$ can be calculated by a hazard-free DNF formula of size at most $\frac{n^{\log n}}{\log n ^{(1-o(1))\log n}}\cdot 2^n$.
\end{lemma}
\begin{proof}
By \Cref{lem:weight-of-random-pi-f}, and the union bound, we get:
\begin{align*}
&\Pr[\exists p_1\in \implicants{f}, |p_1|_{\u}>\log n]\leq \sum_{\log{n}<k\leq n}\Pr[\exists p_1\in \implicants{f}, |p_1|_{\u}=k]\\
&\leq\sum_{\log{n}<k\leq n}{\binom{n}{k}\cdot 2^{n-k}\cdot\left(\frac{1}{2}\right)^{2^k}}< 2^n \cdot \sum_{\log{n}<k\leq n}{\binom{n}{k}\left(\frac{1}{2}\right)^{2^k}}<2^n \cdot \sum_{\log{n}<k\leq n}{n^k\left(\frac{1}{2}\right)^{2^k}}\\
&=\sum_{\log{n}<k\leq n}{\left(\frac{1}{2}\right)^{2^k-(k\log{n}+n)}}\leq\footnotemark n\cdot\left(\frac{1}{2}\right)^{2^{\log{n}+1}-((\log{n}+1)\log{n}+n)}=\frac{n}{2^{\Omega(n)}}.
\end{align*}
\footnotetext{Since $g(k)=2^k-(k\log{n}+n)$ is monotone increasing in k.}

Hence, with high probability, for every $p_1\in \implicants{f}$ we have $|p_1|_{\u}\leq\log n$.
The following DNF formula is a hazard-free formula for $f$:
\[
f(x)=\bigvee_{p_1 \in \implicants{f}}\left(\bigwedge_{(p_1)_i=1}{x_i}\land \bigwedge_{(p_1)_i=0}{\neg x_i}\right).
\]
Each term contains at most $n$ literals. The number of prime implicants is upper bounded by:
\[\sum_{k=1}^{\log{n}}{\binom{n}{k}\cdot2^{n-k}}\leq\log n \cdot \binom{n}{\log n} \cdot 2^{n-\log n}\leq \log n 
\cdot \left(\frac{en}{\log n} \right)^{\log{n}}\cdot 2^{n-\log n}= \frac{n^{\log n}}{\log n ^{(1-o(1))\log n}}\cdot \frac{2^n}{n},
\]
where in the last inequality we used the well known estimate $\binom{a}{b}\leq (ea/b)^b$ where $e$ is the basis of the natural log.
Therefore, we get a formula of size $\frac{n^{\log n}}{\log n ^{(1-o(1))\log n}}\cdot 2^n$.
\end{proof}

\subsection{The hazard-derivative of a random Boolean function}\label{subsec:hd-random-function}
An argument similar to the one in the previous section, shows that the hazard-derivatives of a randomly sampled function have low monotone complexity. Hazard-derivatives indicate whether the function’s output changes when perturbing a fixed input. As more bits in the input are perturbed, the probability that the function's output would not change gets smaller, since there are a lot of resolutions. Consequently, the derivative's output is likely to be 1, making the hazard-derivative nearly constant. 

\begin{lemma}[$\implicants{\dfunc{f}{x}}$ for a random $f$ tend to have more unstable bits]\label{lem:weight-of-random-pi-df}
Let $\booleanx{x}{n}$ and let $\ternaryx{q_{1}}{n}$ be fixed such that  $|q_{1}|_{1}=k$. The probability that $q_{1}\in \implicants{\dfunc{f}{x}}$, for a random Boolean function  $\booleanf{f}{n}$, is $(\frac{1}{2})^{2^k-1}$.
\end{lemma}
\begin{proof}
By \Cref{lem:prime-implicants-df-aux}, $q_{1}\in \implicants{\dfunc{f}{x}}$ if and only if $q_1\in \{1,\u\}^n$ and there exists exactly one $z\in R(x+\u\check{q_1})$ such that $f(z)=\neg f(x)$. The probability of this event is $2\cdot (\frac{1}{2})^{|R(x+\u\check{q_1})|}=2\cdot (\frac{1}{2})^{2^k}=(\frac{1}{2})^{2^k-1}$.
\end{proof}

\begin{restatable}{lemma}{lemrandomder}\label{lem:random-der}
A random Boolean function $\booleanf{f}{n}$, sampled uniformly at random, 
satisfies that with high probability, for every $\booleanx{x}{n}$:
\[
\size_F^{+}(\dfunc{f}{x})\le\frac{n^{\log n}}{\log n ^{(1-o(1))\log n}}.
\]
\end{restatable}
\begin{proof}
We omit some calculations as they closely follow the proof of \Cref{lem:dnf-for-random-f}. Fix $\booleanx{x}{n}$. Taking the union bound over all $\ternaryx{q_1}{n}$ with $1$-weight of at least $\log{n}$ and using \Cref{lem:weight-of-random-pi-df}, we get:
\begin{align*}
&\Pr[\exists q_1\in \implicants{\dfunc{f}{x}}, |q_1|_1>\log{n}]\le\sum_{\log{n}<k\leq n}{\Pr[\exists q_1\in \implicants{\dfunc{f}{x}}, |q_1|_1=k]}\leq\\
&\leq\sum_{\log{n}<k\leq n}{ \binom{n}{k}\left(\frac{1}{2}\right)^{2^k-1}}.
\end{align*}
Taking the union bound over all Boolean inputs of length $n$, we have:
\begin{align*}
&\Pr[\exists x, q_1\in \implicants{\dfunc{f}{x}}, |q_1|_1>\log{n}]\leq 2^n \cdot\sum_{\log{n}<k\leq n}{ \binom{n}{k}\left(\frac{1}{2}\right)^{2^k-1}}=\frac{n}{2^{\Omega(n)}}.
\end{align*}

We have that with high probability, all prime implicants of all hazard-derivatives of $f$ are of $1$-weight at most $\log{n}$.
Fix $\booleanx{x}{n}$. It holds that: 
\[
\df{f}(x;y)=\bigvee_{p_1 \in \implicants{\dfunc{f}{x}}}\bigwedge_{(p_1)_i=1}{y_i}.
\]
This is a monotone DNF formula for $\dfunc{f}{x}$, where each term is a prime implicant. Since the weight of all prime implicants is at most $\log{n}$, each term contains at most $\log{n}$ literals. The number of prime implicants is upper bounded by:
\[\sum_{k=1}^{\log{n}}{\binom{n}{k}}\leq \frac{n^{\log n}}{\log n ^{(1-o(1))\log n}}.
\]
Hence, we get a monotone formula of size $\frac{n^{\log n}}{\log n ^{(1-o(1))\log n}}$.
\end{proof}

The proof of \Cref{prop:random-function-gap}, which we restate for ease of reading, follows immediately.
\proprandomfunctiongap*
\begin{proof}
By \Cref{thm:almost-all-formulas-are-large}, for the majority of Boolean functions it holds that 
$\size_{F}^{\u}(f)\ge \size_{F}(f)=\Omega(\frac{2^n}{\log n})$. 
Combining this with \Cref{lem:random-der}, we concludes the proof.
\end{proof}

\subsection{Hazard-derivatives upper bound}\label{subsec:hazard-derivative-upper-bound}
In this section, we prove \Cref{thm:derivative-upper-bound}, which upper bounds the hazard-free formula size in terms of the monotone formula size of several of its hazard-derivatives.
Building on the analysis of $\matkwu{f}{}$ discussed in \Cref{sec:kw-games-of-the-derivative}, we prove that $\matkwu{f}{}$ can be decomposed into submatrices, each with communication complexity bounded from above by that of a specific hazard-derivative. When given a subset $\{x^{(1)},\dots,x^{(k)}\}$ of boolean inputs in $X_b\subseteq \{0,1\}^n$ for some $b\in\{0,1\}$ that are enough to cover $P_b$ (i.e., for every $p\in P_b$ there exists $x\in X_b$ such that $p\precequ x$), we decompose $\matkwu{f}{}$ to $k$ (potentially overlapping) $M_{x}$'s with labels $\implic{f}{\neg b}\times\implic{f}{b}|_x$, as illustrated in \Cref{fig:matrics-decomposition-demonstrations} (see \Cref{fig:example-f} for a concrete example).

We leverage \Cref{thm:derivative-upper-bound} to prove \Cref{thm:ub-tight}. Thus, we provide an alternative proof for the improved upper bound on the hazard-free formula complexity of random Boolean functions (see \Cref{thm:universal-ub}).
Furthermore, \Cref{thm:ub-tight} demonstrates that the upper bound in \Cref{thm:derivative-upper-bound} is not too far from the truth for random Boolean functions. In \Cref{subsection:tight-lower-bound-range} we provide an example of an explicit function for which \Cref{thm:derivative-upper-bound} yields a tight bound. 

To ease the reading, we repeat the statement of \Cref{thm:derivative-upper-bound}.

\thmderivativesupperbound*

\begin{figure}[t]
\centering
\begin{tikzpicture}[overlay]
    \node[font=\small] at (-2.75,-1.1) {$P_{\neg b}^{(f)}$};
    
    \draw[decorate, decoration={brace, mirror}] (-2.12,-2.25) -- node[below=1pt, font=\tiny] {$\implic{f}{b}|_{x^{(1)}}$} (-0.37,-2.25);

    \draw[decorate, decoration={brace}] (-2.11,0) -- node[above=1pt, font=\tiny] {$\implic{f}{b}|_{x^{(k)}}$} (-2.11+0.4,0);
    
    \draw[decorate, decoration={brace}] (-0.72,0) -- node[above=1pt, font=\tiny] {$\implic{f}{b}|_{x^{(2)}}$} (0.35,0);

    \draw[decorate, decoration={brace}] (1.02,0) -- node[above=1pt, font=\tiny] {$\implic{f}{b}|_{x^{(k)}}$} (2.1,0);
\end{tikzpicture}
\begin{minipage}{\textwidth}
$$
\begin{pNiceArray}[margin]{ccccccccccc}
  \Block[tikz={fill=red!50}]{5-5}<\small>{M_{x^{(1)}}} & & & &\Block[tikz={opacity=0.7, fill=teal!50}]{5-3}<\small>{M_{x^{(2)}}} & & & & \Block[tikz={opacity=0.5, fill=black!50}]{5-3}<\small>{M_{x^{(k)}}} & & \\
  &&&&&&&&&& \\
  &&&&&&& \dots &&& \\
  &&&&&&&&&& \\
  &&&&&&&&&& \\
\end{pNiceArray}
$$
\vspace{2pt}
\end{minipage}

\begin{tikzpicture}[overlay]
\draw[opacity=0.5, fill=black!50, draw=none] (-2.07,2.33) rectangle node[font=\small]{} (-1.67,0.22);
\end{tikzpicture}
\vspace{5pt}
\caption{Decomposition of $\matkwu{f}{}$.}
\Description{A sketch of a decomposition of $\matkwu{f}{}$ to $k$ (potentially overlapping) $M_{x}$'s with labels $\implic{f}{\neg b}\times\implic{f}{b}|_x$.}
\label{fig:matrics-decomposition-demonstrations}
\end{figure}

\begin{proof}
All hazard-derivatives are monotone Boolean functions (\Cref{lem:derivative-monotone}), so following \Cref{thm:kw-hazard-free-monotone}, it is enough to prove:
\[
\size_{F}^\u(f)\le \sum_{x\in X_b}\size_{F}^{\u}(\dfunc{f}{x}).
\]
We design a protocol $P$ for $\kwgame_f^\u$ that has at most $\sum_{x\in X_b}\monorect(\matkwu{\dfunc{f}{x}}{})$ many leaves.

W.l.o.g, assume that $b=0$, so the column labels of $\matkwu{f}{}$  are $\implic{f}{b}$. We denote $X_b=\{x^{(1)},\dots, x^{(k)}\}$ and partition $\implic{f}{b}$ to $\le k$ distinct subsets $A_1,\dots, A_k \subseteq \implic{f}{b}$ according to $X_b$:
\[
A_1:=\implic{f}{b}|_{x^{(1)}},
\]
\[
\forall i\in [k], i>1, \quad A_i:=\implic{f}{b}|_{x^{(i)}} \setminus \bigcup_{j=1}^{i-1}A_{j}.
\]
It is easy to see that: 
\[
\implic{f}{b}=\bigcup_{i=1}^kA_{i}, \quad \text{and} \quad  \forall i\in [k], \quad A_i\subseteq \implic{f}{b}|_{x^{(i)}}.
\]
We denote by $M_i$ the submatrices of $\matkwu{f}{}$ with labels $\implic{f}{\neg b} \times A_i$. We also denote by $M_{x^{(i)}}$ the submatrices of $\matkwu{f}{}$ with labels $\implic{f}{\neg b} \times \implic{f}{b}|_{x^{(i)}}$. Consider the protocol $P$ which begins by partitioning $\matkwu{f}{}$ into the distinct submatrices $M_1,\dots, M_k$, we get:
\[
\monorect(\matkwu{f}{})\le \sum_{i=1}^{k}\monorect(M_i).
\]
Since $M_i$ is a submatrix of $M_{x^{(i)}}$, it follows from \Cref{prop:submatrix-of-derive-p-equiv-derivative} that 
\[
\monorect(M_i)\le \monorect(M_{x^{(i)}})=\monorect(\matkwu{\dfunc{f}{x^{(i)}}}{}),
\]
and the claim follows.
\end{proof}

As a corollary, we get a universal upper bound on the hazard-free formula complexity of random Boolean functions, which in particular proves \Cref{thm:universal-ub}.

\begin{theorem}[A universal upper bound for most Boolean functions]
    \label{thm:ub-tight}
For a random Boolean function $f$, we get that with high probability, for a minimally sized $X$ such that 
\[
\implicates{f}=\bigcup_{x\in X}\implicates{f}|_x,
\]
the following holds:
\[
\left(\sum_{x\in X}\size_{F}^{+}(\dfunc{f}{x})\right)^{1-o(1)}\leq \size_{F}^\u(f)\le \sum_{x\in X}\size_{F}^{+}(\dfunc{f}{x}).
\]
In particular, for a random Boolean function $f$, with high probability, we get that:
\[
\frac{2^n}{\log n}\leq \size_{F}^\u(f)  \leq 2^n\cdot \frac{n^{\log n}}{\log n ^{(1-o(1))\log n}}.
\]
\end{theorem}
\begin{proof}
The claims follow immediately from \Cref{thm:almost-all-formulas-are-large}, \Cref{lem:random-der} and \Cref{thm:derivative-upper-bound}. 
\end{proof}

\subsubsection{Tight bound for range functions}\label{subsection:tight-lower-bound-range}
Here, we prove that \Cref{thm:derivative-upper-bound} can give exact bounds. We consider the special case where there exists $\booleanx{x^*}{n}$ such that $\implicates{f}=\implicates{f}|_{x^*}\cup \implicates{f}|_{\neg x^*}$. The key observation is that in this case, the prime implicates are covered by $x^*$ and $\neg x^*$. Since $x^*$ and $\neg x^*$ are negated, all stable coordinates of $\implic{f}{0}|_{x^*}$ and $\implic{f}{0}|_{\neg x^*}$ are different, making the row entries of $M_{x^*}$ and $M_{\neg x^*}$ disjoint. Therefore, there is an optimal protocol that begins by dividing the communication matrix to $M_{x^*}$ and $M_{\neg x^*}$.
Later, we discuss a family of explicit functions that satisfy the above property.

\begin{proposition}\label{prop:lower-bound-partition-x}
Let $\booleanf{f}{n}$ be a non-constant Boolean function. If there exists $\booleanx{x^*}{n}$ such that $\implicates{f}=\implicates{f}|_{x^*}\cup \implicates{f}|_{\neg x^*}$, then:
\[
\size_{F}^\u(f)=\size_{F}^+(\dfunc{f}{x^*})+\size_{F}^+(\dfunc{f}{\neg x^*}).
\]
\end{proposition}
\begin{proof}
Let $P$ be an optimal protocol for $\kwgame_f^\u$ and let $M:=\matkwu{f}{}$. Let $\rec=A\times B$ be a $\kwgame_f^\u$-monochromatic rectangle in the partition induced by $P$ (in particular, $A\subseteq \implicants{f}$ and $B\subseteq \implicates{f}$).
We first prove that $\rec\subseteq \implicants{f}\times \implicates{f}|_{x^*}$ or $\rec\subseteq \implicants{f}\times \implicates{f}|_{\neg x^*}$.
Indeed, assume towards contradiction that $B\cap \implicates{f}|_{x^*} \neq \emptyset$ and $B\cap \implicates{f}|_{\neg x^*}\neq\emptyset$. Hence, there exist $p_1\in \implicants{f}$, $p_0\in \implicates{f}|_{x^*}$ and $q_0 \in \implicates{f}|_{\neg x^*}$ such that $(p_1,p_0),(p_1,q_0)\in \rec$. Since $\rec$ is $\kwgame_f^\u$-monochromatic there must exist $d\in [n]$ such that $d\in M_{p_1,p_0} \cap M_{p_1,q_0}$. Meaning $(p_1)_d\neq (p_0)_d$ and $(p_1)_d\neq (q_0)_d$ and $(p_1)_d,(p_0)_d,(q_0)_d$ are stable. Therefore, $(p_0)_d=(q_0)_d$. This is a contradiction, since $p_0$ derives $x^*$ and $q_0$ derives $\neg x^*$, their stable coordinates must be different.

Next, we describe a protocol $P'$ for $\kwgame_f^\u$ that begins with partitioning $\implicates{f}$ to $\implicates{f}|_{x^*}$ and $\implicates{f}|_{\neg x^*}$, without increasing the number of leaves. Bob starts by partitioning $\implicates{f}$ to $\implicates{f}|_{x^*}$ and $\implicates{f}|_{\neg x^*}$, this step requires $1$ bit of communication. Then, Alice and Bob follow $P$ on the inputs in $\implicants{f}\times \implicates{f}|_{x^*}$ or $\implicants{f}\times \implicates{f}|_{\neg x^*}$, according to the first bit, and return the obtained answer.

We denote by $P_{x^*}$ the protocol $P$ played only on inputs in $\implicants{f}\times \implicates{f}|_{x^*}$, and by $P_{\neg x^*}$ the protocol $P$ played only on inputs in $\implicants{f}\times \implicates{f}|_{\neg x^*}$. Since every monochromatic rectangle in the partition induced by $P$ is contained in exactly one of these inputs subsets, we can eliminate leaves in $P_{x^*}$ and in $P_{\neg x^*}$ that the players would never reach and get:

\begin{equation}\label{eq:exact-bound-range-a}
\monorect(P')=\monorect(P_{x^*})+\monorect(P_{\neg x^*})=\monorect(P).
\end{equation}

Notice that $P_{x^*}$ and $P_{\neg x^*}$ are protocols for the submatrices of $\matkwu{f}{}$ with labels $\implicants{f}\times \implicates{f}|_{x^*}$ and $\implicants{f}\times \implicates{f}|_{\neg x^*}$, respectively. By \Cref{prop:submatrix-of-derive-p-equiv-derivative}, these submatrices are equivalent in terms of communication complexity to $\matkwu{\dfunc{f}{x^*}}{}$ and $\matkwu{\dfunc{f}{\neg x^*}}{}$. According to \Cref{thm:kw-hazard-free-monotone} and \Cref{thm:kw-monotone}, we get:
\begin{gather}
\monorect(P_{x^*})\ge \monorect(\matkwu{\dfunc{f}{x^*}}{})=\size_F^+(\dfunc{f}{x^*}), \label{eq:exact-bound-range-b} \\ \monorect(P_{\neg x^*}) \ge \monorect(\matkwu{\dfunc{f}{\neg x^*}}{}) =\size_F^+(\dfunc{f}{\neg x^*}). \label{eq:exact-bound-range-c}
\end{gather}
Following \eqref{eq:exact-bound-range-a}, \eqref{eq:exact-bound-range-b} and \eqref{eq:exact-bound-range-c}, we have:

\begin{equation}\label{eq:exact-bound-range-d}
\size_F^\u(f)=\monorect(P')\ge \size_F^+(\dfunc{f}{x^*})+ \size_F^+(\dfunc{f}{\neg x^*}).
\end{equation}

\Cref{thm:derivative-upper-bound} implies that the lower bound in \eqref{eq:exact-bound-range-d} is tight.
\end{proof}

We next present an example of an explicit function that satisfies the assumption of \Cref{prop:lower-bound-partition-x}.

\begin{definition}[Threshold and Range functions]\label{def:range-functions}
Let $k\in [n]$. A \emph{threshold-$k$} function $\booleanf{T_k^n}{n}$ is defined such that $T_k^n=1$ iff $x_1+\dots+x_n\ge k$.

Let $a,b\in [n]$ such that $0<a<b\le n$. We call $R_{a,b}^n$ an \emph{$a$-$b$-range function}, if for every $\booleanx{x}{n}$, $R_{a,b}^n(x)=1$ iff $a\leq x_1+\dots+x_n<b$. In other words, $R_{a,b}^n(x)=T^n_a(x) \land \neg T^n_b(x)$.
\end{definition}

\begin{lemma}[Prime implicants and prime implicates of range functions]\label{lem:prime-of-range-functions}
Let $\booleanf{R_{a,b}^n}{n}$ be an $a$-$b$-range function. Let $\ternaryx{p}{n}$, then:
\begin{enumerate}
    \item $p\in \implicants{R_{a,b}^n}$ if and only if $|p|_1=a$ and $|p|_0=n-(b-1)$. \label{item:prime-of-range-functions-1}
    \item $p\in \implicates{R_{a,b}^n}$ if and only if $|p|_0=0$ and $|p|_1=b$ \emph{or} $|p|_1=0$ and $|p|_0=n-(a-1)$. \label{item:prime-of-range-functions-2}
\end{enumerate}
\end{lemma}
\begin{proof}
We prove Property \ref{item:prime-of-range-functions-2}, since Property \ref{item:prime-of-range-functions-1} can be proved similarly.

It is clear that if an input satisfies Property \ref{item:prime-of-range-functions-2} then it is an implicate. To prove that these are the only prime implicates consider $p\in \implicates{R_{a,b}^n}$. In order for $p$ to be an implicate of $R_{a,b}^n$, every resolution must have at most $a-1$ $1$s or at least $b$ $1$s. Therefore, we have to set at least $n-(a-1)$ coordinates of $p$ to $0$ or at least $b$ coordinates of $p$ to $1$. Naturally, the minimal implicates with respect to $\precequ$ are those with $|p|_0=0$ and $|p|_1=b$ \emph{or} $|p|_1=0$ and $|p|_0=n-(a-1)$. 
\end{proof}

\begin{lemma}\label{lem:derivatives-of-range-functions}
Let $\booleanf{R_{a,b}^n}{n}$ be an $a$-$b$-range function, then:
\[
\dfunc{R_{a,b}^n}{\bar{0}}=T_{a}^n,
\]
\[
\dfunc{R_{a,b}^n}{\bar{1}}=T_{n-(b-1)}^n.
\]
\end{lemma}
\begin{proof}
Note that $R_{a,b}^n(\bar{0})=0$ and $R_{a,b}^n(\bar{1})=0$. Considering \Cref{lem:prime-of-range-functions}, we calculate the prime implicates of $\dfunc{R_{a,b}^n}{\bar{0}}$.
Let $p\in \implicates{R_{a,b}^n}$. Following \Cref{lem:prime-implicates-df}, $p$ derives $\bar{0}$ iff $|p|_1=0, |p|_0=n-(a-1)$. Therefore, $\ternaryx{q_0}{n}$ is a prime implicate of $\dfunc{R_{a,b}^n}{\bar{0}}$ iff $|q_0|_0=n-(a-1)$. It is not difficult to prove that those are the prime implicates of $T_a^n$. 

Similarly, the prime implicates of $\dfunc{R_{a,b}^n}{\bar{1}}$ have exactly $b$ $0$s. 
\end{proof}

\begin{proposition}[Tight lower bound for range functions]\label{prop:tight-lower-bound-range}
Let $\booleanf{f}{n}$ be an $a$-$b$-range function, then:
\[
\size_F^\u(f)=\size_F^+(T_a^n)+\size_F^+(T_{n-(b-1)}^n).
\]
\end{proposition}
\begin{proof}
 \Cref{lem:prime-of-range-functions} implies the following for $x^*:=\bar{0}$:
\[
\implicates{R_{a,b}^n}=\implicates{R_{a,b}^n}|_{x^*} \cup \implicates{R_{a,b}^n}|_{\neg x^*}.
\]
We derive the claim following \Cref{prop:lower-bound-partition-x} and \Cref{lem:derivatives-of-range-functions}.
\end{proof}

\section{Composition of Boolean functions in hazard-free setting}\label{sec:composition-hazard-free}
In this section, we prove \Cref{prop:lower-bound-direct-sum}, which establishes a lower bound on the communication complexity of the hazard-free KW game for the composition of a monotone function with an arbitrary inner function. We follow the proof idea of the next lemma:
\begin{lemma}[{\cite[Lemma 4]{KRW95}}]\label{lem:composition-of-monotone-relations}
Let $\booleanf{f}{n}$ and $\booleanf{g}{m}$ be monotone Boolean functions. Then:
\[
\kwgame_{f}^+\otimes\kwgame_{g}^+\leq \kwgame^+_{f\diamond g}.
\]
\end{lemma}
The proof of \Cref{lem:composition-of-monotone-relations} uses two properties of $\kwgame^+_{f\diamond g}$: First, the game can be played equivalently on the prime implicants and the prime implicates of the function, and those of $f\diamond g$ are constructed of those of $f$ and those of $g$. Second, for an input $\ternaryx{X,Y}{n\times m}$, a valid answer $(i,j)$ to $\kwgame^+_{f\diamond g}$ must be a row $i$ for which $g(X_i)=1, g(Y_i)=0$ and so $i$ is a valid answer to the KW game of $f$ and $j$ is a valid answer to the KW game of $g$.
The lower bound is then derived through a reduction to the direct sum of functions, as shown in \cite[Corollary 5]{KRW95}. We prove that these properties also hold in the hazard-free KW game of $f\diamond g$ if $f$ is a monotone function, while $g$ can be any function. 

We prove in \Cref{prop:prime-implicants-composition} that the prime implicants (and implicates) of the composed function, $f\diamond g$, can be derived from the prime implicants (and implicates) of $f$ and $g$. 

\begin{proposition}[Prime implicants of composition]\label{prop:prime-implicants-composition}
Let $\booleanf{f}{n}$ and $\booleanf{g}{m}$ be Boolean functions. Then $P\in \implicants{f\diamond g} \subseteq \{0,\u,1\}^{nm}$ if and only if:
\begin{enumerate}
    \item $\tilde{g}[P] \in \implicants{f}$. \label{item:prime-implicants-composition-1}
    \item For every $i\in[n]$, if $\tilde{g}(P_i)=b\in\{0,1\}$ then $P_i\in P_{b}^{(g)}$. Otherwise, $P_i=\u^m$. \label{item:prime-implicants-composition-2}
\end{enumerate}
\end{proposition}
In order to prove \Cref{prop:prime-implicants-composition} we would need the following lemma, which shows that the hazard-free extension of the block composition of two Boolean functions is the same as the block composition of their hazard-free extensions. We define the block composition of ternary functions exactly as in the Boolean case.
\begin{restatable}[Hazard-free extension of composition]{lemma}{lemhazardcompositionnatural}\label{lem:hazard-free-extension-composition}
Let $\booleanf{f}{n}$ and $\booleanf{g}{m}$ be Boolean functions. Then:
\[
\widetilde{f\diamond g} = \tilde{f} \diamond \tilde{g}.
\]
\end{restatable}
We begin by proving \Cref{prop:prime-implicants-composition} and later prove \Cref{lem:hazard-free-extension-composition}.
\begin{proof}[Proof of \Cref{prop:prime-implicants-composition}]
Following \Cref{lem:hazard-free-extension-composition}, we refer to $\tilde{f}\diamond\tilde{g}$ and $\widetilde{f\diamond g}$ interchangeably.

$\subseteq:$ Let $P\in \implicants{f\diamond g}$. Clearly, $\widetilde{f\diamond g}(P)=\tilde{f}\diamond\tilde{g}(P)=\tilde{f}(\tilde{g}[P])=1$. Condition \ref{item:prime-implicants-composition-1} now follows from \Cref{def:prime-implicant}.

Let $i\in [n]$. We prove that condition \ref{item:prime-implicants-composition-2} is satisfied. Let $j\in [m]$ such that $P_{i,j}$ is stable. $P$ is a prime implicant of $f\diamond g$, hence, $\tilde{f}\diamond\tilde{g}(\replace{P}{i,j}{\u})=\u$. We conclude that $\tilde{g}(P_i)=b\in\{0,1\}$ and $\tilde{g}(\replace{P_i}{j}{\u})=\u$, otherwise, the output of $\tilde{f}\diamond\tilde{g}$ would not change. Since $P_{i,j}$ is stable if and only if $(P_i)_j$ is stable, we get that every $P_i$ with some stable coordinate is a prime implicant (if $\tilde{g}(P_i)=1$) or implicate (if $\tilde{g}(P_i)=0$) of $g$. 

$\supseteq:$ Let $P\in \{0,\u,1\}^{nm}$  that satisfies conditions \ref{item:prime-implicants-composition-1} and \ref{item:prime-implicants-composition-2}. Condition \ref{item:prime-implicants-composition-1} implies that $\tilde{f}\diamond \tilde{g}(P)=\tilde{f}(\tilde{g}[P])=1$, and so $P$ is an implicant of $f\diamond g$. Let $(i,j)\in [n] \times [m]$ such that $P_{i,j}$ is stable. Denote $P'=\replace{P}{(i,j)}{\u}$. We conclude from condition \ref{item:prime-implicants-composition-2} that there exists $b\in\{0,1\}$ such that $P_i\in P_{b}^{(g)}$. Therefore, $\tilde{g}(P'_i)=\u$ and $\tilde{f}\diamond \tilde{g}(P')=\tilde{f}(\tilde{g}[P'])=\tilde{f}(\replace{\tilde{g}[P]}{i}{\u})=\u$ since $\tilde{g}[P]$ is a prime implicant of $f$.
\end{proof}

For the proof of \Cref{lem:hazard-free-extension-composition} we require the next simple lemma.

\begin{restatable}{lemma}{lemhazardfreeextension}\label{lem:hazard-free-extension-composition-aux}
Let $\booleanf{g}{m}$ be a ternary function. Let $\ternaryx{P}{nm}$ be a ternary $n\times m$ matrix. Then:
\[
R(\tilde{g}[P])=\{g[Z]: Z\in R(P)\}.
\]    
\end{restatable}
\begin{proof}
$\subseteq:$ Let $z\in R(\tilde{g}[P])=R((\tilde{g}(P_1),\dots,\tilde{g}(P_n)))$. Let $i\in [n]$. We show that there exists $Z_i\in R(P_i)$ such that $g(Z_i)=z_i$. Indeed, if $\tilde{g}(P_i)=z_i$, we can choose any $Z_i\in R(P_i)$.
Otherwise, $\tilde{g}(P_i)=\u$. Therefore, we can choose $Z_i\in R(P_i)$ such that $g(Z_i)=z_i$.
Hence, we get that the matrix $Z$, whose rows are $Z_1,\dots,Z_n$, satisfies that $Z\in R(P)$ and $g[Z]=z$ as required.

$\supseteq:$ Let $Z\in R(P)$. Assume towards contradiction that $g[Z]\not\in R(\tilde{g}[P])$. Since $g[Z]$ is Boolean, there must exist some $b\in \{0,1\}$ and $i\in [n]$ such that $(\tilde{g}[P])_i=\tilde{g}(P_i)=b$ and $(g[Z])_i=g(Z_i)=\neg b$. As $Z_i\in R(P_i)$, by \Cref{def:hazard-free-extension} we have that $g(Z_i)=b$, a contradiction.
\end{proof}

\begin{proof}[Proof of \Cref{lem:hazard-free-extension-composition}]
Let $\ternaryx{P}{nm}$ be and $b\in \{0,1\}$. By \Cref{def:hazard-free-extension} and \Cref{lem:hazard-free-extension-composition-aux}, we have:
\begin{align*}
&\widetilde{f\diamond g}(P)=b \iff \\
    &\forall Z\in R(P), \quad f\diamond g(Z)=b \iff \\
    &\forall Z\in R(P), \quad f(g[Z])=b \iff \\
    &\forall z\in R(\tilde{g}[P]), \quad f(z)=b \iff \\ 
    &\tilde{f}(\tilde{g}[P])=b \iff \\
    &\tilde{f}\diamond \tilde{g}(P)=b. \qedhere
\end{align*}
\end{proof}

Given \Cref{prop:prime-implicants-composition}, if $\implicants{f} \subseteq \{1,\u\}^n$ and $\implicates{f} \subseteq \{0,\u\}^n$, a valid answer to the hazard-free KW game of $f\diamond g$ on the inputs $\ternaryx{X,Y}{n m}$ must lie in a row $i$ such that $X_i\in \implicants{g}$ and $Y_i \in \implicates{g}$. Hence, we have $\tilde{g}(X_i)=1$ and $\tilde{g}(Y_i)=0$, without requiring any additional properties from $g$. 
This is exactly the case if $f$ is monotone. Therefore, we obtain an analog of \cite[Lemma 4]{KRW95}:
\begin{restatable}[Direct sum is reducible to $\kwgame_{f\diamond g}^\u$]{lemma}{lemreductionofdirectsumtocomp}\label{lemma:reduction-of-direct-sum-to-comp}
    Let $\booleanf{f}{n}$ be a \emph{monotone} Boolean function, and let $\booleanf{g}{m}$ be a Boolean function. Then, it holds that:
    \[
    \kwgame_{f}^{+} \otimes \kwgame_{g}^{\u}\le \kwgame_{f\diamond g}^\u.
    \]
\end{restatable}

\begin{proof}
Let $\phi:\{0,\u,1\}^n\times \{0,\u,1\}^m \rightarrow \{0,\u,1\}^{nm}$ be the following function: given $p_f\in \{0,\u,1\}^n$ and $p_g\in \{0,\u,1\}^m$, $\phi(p_f,p_g)$ returns $P\in \{0,\u,1\}^{nm}$ such that:
\[
\forall i\in[n], \quad P_i:=\begin{cases}
    p_g, & (p_f)_i\neq \u, \\
    \u^m, & otherwise.
\end{cases}
\]
\newline Let $\psi:[nm]\rightarrow [n]\times[m]$ be the following function: given $k\in[nm]$, $\psi(k)=(\left\lfloor \frac{k-1}{m} \right\rfloor+1, k -m\cdot (\left\lfloor \frac{k-1}{m} \right\rfloor))$.
We next show that every $k\in[nm]$, $(p_f^{(1)}, p_g^{(1)}) \in \implicants{f} \times \implicants{g}$ and $(p_f^{(0)}, p_g^{(0)})\in \implicates{f} \times \implicates{g}$ satisfy:
\[
(\phi(p_f^{(1)}, p_g^{(1)}), \phi(p_f^{(0)}, p_g^{(0)}), k) \in \kwgame_{f\diamond g}^\u \Rightarrow ((p_f^{(1)}, p_g^{(1)}), (p_f^{(0)}, p_g^{(0)}),\psi(k)) \in \kwgame_{f}^{+} \otimes \kwgame_{g}^{\u}.
\]
\Cref{cor:prime-monotone} implies that $p_f^{(1)}\in\{1,\u\}^n$. From 
\Cref{prop:prime-implicants-composition}, it follows that $\phi(p_f^{(1)}, p_g^{(1)})\in \implicants{f\diamond g}$. Similarly, $\phi(p_f^{(0)}, p_g^{(0)})\in\implicates{f\diamond g}$. Let $k\in[nm]$ be a valid answer for $\kwgame^\u_{f\diamond g}$ on the aforementioned inputs. From the definition of hazard-free $KW$ game, $k$ is a coordinate where $\phi(p_f^{(1)}, p_g^{(1)})$ and $\phi(p_f^{(0)}, p_g^{(0)})$ are both \emph{stable} and \emph{different}. Hence, from the definition of $\phi$, $i=\left\lfloor \frac{k-1}{m} \right\rfloor+1$ is an index of a row  that satisfies $1=(p_f^{(1)})_i\neq (p_f^{(0)})_i=0$. Moreover, the entry corresponding to $k$ in the $i$'th row is $j=k- m(i-1)$, and it must be the case that $(p_g^{(1)})_j\neq (p_g^{(0)})_j$ and they are both stable. Therefore, $\psi(k)$ is a valid answer for $\kwgame_{f}^{+} \otimes \kwgame_{g}^{\u}$, as required.
\end{proof}

We are now ready to prove \Cref{prop:lower-bound-direct-sum}, which we restate to ease the reading.

\proplowerboundviadirectsum*
\begin{proof}
The claim follows from \Cref{lemma:reduction-of-direct-sum-to-comp} and the rank lower bound method in communication complexity (see \cite[Lemma 1.28]{KN97}).
\end{proof}

\Cref{prop:lower-bound-direct-sum} provides a step toward advancing hazard-free lower bound techniques, and we hope it encourages further research in this direction. As an application of \Cref{prop:lower-bound-direct-sum}, we prove a lower bound on the hazard-free formula depth of the set covering function composed with the multiplexer function. The set covering function is also used in \cite[Theorem 13]{KRW95}, where it is composed with itself.
\begin{definition}[Set covering $MC_{k,n}$ {\cite[Example 5.2.1]{Kar89}}]\label{def:set-covering}
Given a bipartite graph $G=(U\cup V,E)$ with $|U|=|V|=n$ we have $MC_{k,n}(G)=1 \iff$ there exists a $U'\subseteq U$ such that $|U'|=k$ and every node in $V$ has a neighbor in $U'$.
\end{definition}

It is not hard to see that $MC_{k,n}$ is monotone (in the edge set of $G$). Moreover, the disjointness  function reduces to $\kwgame_{MC_{k,n}}^+$.
\begin{definition}[The disjointness function]\label{def:disjointness-function}
Let $\mathbf{P}([n])$ denote all subsets of $[n]$, and let $\mathbf{P}_l([n])$ denote all subsets of $[n]$ of size $l$:

\begin{itemize}
    \item Let $I_n:\mathbf{P}([n]) \times \mathbf{P}([n]) \rightarrow \{0,1\}$ where $I_n(S, T) = 0 \iff S \cap T = \emptyset$.
    \item For $l \le n/2$, let $I_{l,n}: \mathbf{P}_l([n]) \times \mathbf{P}_l([n]) \rightarrow \{0,1\}$ where $I_{l,n}(S, T) = 0 \iff S \cap T = \emptyset$.
\end{itemize}
\end{definition}

\begin{theorem}[$I_{k,n}$ is reducible to set covering {\cite{Raz90}, \cite[Theorem 5.2.1]{Kar89}}]\label{thm:disjointness-to-set-covering}
Let $k=c\log n$ for some suitable $c > 0$. Then $I_{k,n}\le \kwgame^+_{MC_{k,n}}$. 
\end{theorem}

\cite{IK23} proved that the $\text{subcube-intersect}_n$ reduces to $\kwgame_{\mux_n}^\u$.

\begin{definition}[subcube-intersect {\cite[p.15]{IK23}}]
\sloppy Let $\text{subcube-intersect}_n:\{0,\u,1\}^n \times \{0,\u,1\}^n\rightarrow \{0,1\}$ be a \emph{function} such that $\text{subcube-intersect}(x,y)=1$ if and only if $x$ and $y$ have a common resolution.
\end{definition}

\begin{lemma}[{\cite[Lemma 5.10]{IK23}}]\label{lem:reduction-sub-inter-mux}
The $\text{subcube-intersect}_n$ communication problem reduces to the communication problem $\kwgame_{\mux_n}^\u$ with no extra cost.
\end{lemma}

Since both $I_{l,n}$ and $\text{subcube-intersect}_n$ are \emph{functions}, the communication complexity of the associated problem can be lower bounded by the rank of the corresponding matrices.
\begin{theorem}[\cite{kan72}]\label{thm:rank-of-disjointness}
\[
\rank(M_{I_n})=2^n,\rank(M_{I_{l,n}})=\binom{n}{l}.
\]
\end{theorem}

\begin{lemma}[{\cite[Lemma 5.12]{IK23}}]\label{lem:rank-subcube-intersect}
$\rank(M_{\text{subcube-intersect}_n})=3^n$ for every $n\ge1$.
\end{lemma}

\begin{proposition}
    [Composition of the set covering and multiplexer functions]\label{prop:comp-monotone-mux}
Let $\mux_m$ be the multiplexer function as in \Cref{def:mux}. Let $MC_{k,n}$ be the set covering function as in \Cref{def:set-covering}, such that $k=c\log n$ for some suitable constant $c>0$ as in \Cref{thm:disjointness-to-set-covering}. Then,
\[
\depth_{F}^\u(MC_{k,n} \diamond \mux_m) \ge \log(3^m)+\log\binom{n}{c \log n} \ge\log(3)\cdot m +\Omega(c(\log n)^2).
\]
\end{proposition}
\begin{proof} 
By \Cref{lem:reduction-sub-inter-mux}, we have $\kwgame_{\mux_m}^\u\ge subcube-intersect_m$.  \Cref{thm:disjointness-to-set-covering} shows that $MC_{k,n}\ge I_{k,n}$. Following \Cref{prop:lower-bound-direct-sum}, we get:
\[
CC(\kwgame_{MC_{k,n} \diamond \mux_m}^\u) \ge \log(\rank(M_{subcube-intersect_m}))+\log(\rank(M_{I_{k,n}})).
\]
By \Cref{lem:rank-subcube-intersect} and \Cref{thm:rank-of-disjointness}, we get:
\begin{align*}
&\log(\rank(M_{subcube-intersect_m}))+\log(\rank(M_{I_{k,n}})) =\\ &\log(3^m)+\log\binom{n}{k} \ge \log(3)\cdot m+k\cdot\log\left(\frac{n}{k}\right)=\log(3)\cdot m+c\log^2 n-c\log(n)\log(c\log n)=\\
&\log(3)\cdot m + (1+o(1))c\log^2 n.
\end{align*}
The claim now follows from \Cref{thm:kw-hazard-free}.
\end{proof}

\begin{remark}
    As $MC_{k,n}$ is monotone, it equals its hazard-derivative at  $\bar{0}$, which is also the ``hardest'' derivative (this can be inferred from the proof of \Cref{thm:derivative-lower-bound}). Hence, a formula for the hazard-derivative of $MC_{k,n} \diamond \mux_m$ can be obtained from composing the monotone formula for $MC_{k,n}$ with that of the hazard-derivative of $\mux_m$ (see the ``chain rule'' in \cite[Lemma 4.7]{IKL+19}).
    As shown in \cite[Proposition 1]{IK23}, every hazard-derivative of $\mux_m$, has a formula of depth $m+\left\lceil \log (m +1)\right\rceil$. Thus, every hazard-derivative of $MC_{k,n} \diamond \mux_m$ has a monotone formula of depth $\left\lceil\log \binom{n}{k}\right\rceil +\left\lceil \log n\right\rceil+\left\lceil \log k\right\rceil+m+\left\lceil \log (m +1)\right\rceil$, 
    which for our setting of parameters would give a lower bounds of $m + (1+o(1))c\log^2 n $, which is worse than what \Cref{prop:comp-monotone-mux} yields. 
\end{remark}

\begin{remark}
    Since we are using a monotone outer function, the advantage that we can get compared to the hazard-derivative method must come from the inner function.
\end{remark}

\section{Open problems}\label{sec:open}
The most obvious open problem is determining whether the universal hazard-free formula size upper bound in the worst case is equal to that of a random function.

\begin{question}
    Is it true that for every Boolean function $\booleanf{f}{n}$, $\size_{F}^\u (f) \leq 2^{(1+o(1))n}$?
\end{question}

Another interesting question is whether \Cref{thm:kw-f-not-breaks-if} also holds for circuits:

\begin{question}
    Let $\booleanf{f}{n}$ be a non-constant Boolean function. Is it true that there exists $\booleanx{x}{n}$ such that:
\[
    \size_{C}^{+}(\dfunc{f}{x}) = \size_{C}^{\u}(f),
\]
if and only if $f$ is a unate function?
\end{question}

The validity of the KRW conjecture in the hazard-free setting is also an intriguing question. Many previous works on the KRW conjecture have considered restricted versions to gain intuition and develop tools, including \cite{KRW95, hw92, EIRS01, GMWW17, km18, Mei20}, among others. We believe that the hazard-free KRW conjecture presents an interesting restricted model to study, with potential applications to all Boolean functions.

\conjhazardfreekrw*

Finally, we consider block compositions with the $\xor$ function. Recall that the hazard-free complexity of $\xor_n$ is equal to its standard formula complexity. Additionally, its prime implicants and prime implicates are $f^{-1}(1)$ and $f^{-1}(0)$, respectively, meaning its hazard-free and the classic KW games are equivalent. This naturally leads to the question of whether the lower bound proof of Dinur and Meir  \cite{DM18} can be adapted to the hazard-free setting.
Specifically, we wonder if, in the following theorem, one can replace $\size_F$ with $\size_F^{\u}$.

\begin{theorem}[Composition over parity {\cite[Theorem 3.1]{DM18}}]\label{thm:parity-composition}
Let $\booleanf{f}{k}$ be a non-constant Boolean function. Then:
\[
\size_F(f\diamond \xor_m)\ge \frac{\size_F(f)\cdot \size_F(\xor_m)}{2^{\tilde{O}(\sqrt{k+\log m})}},
\]
where $\tilde{O}(t):=O(t\cdot \log^{O(1)}t)$. 
\end{theorem}

\begin{acks}
The first author would like to thank Maya Baruch for her invaluable discussions on various topics, insightful feedback on the proofs, and support in articulating ideas clearly. She is also grateful to Ido Weinstein and Rotem Zluf for their willingness to listen to her early ideas, which contributed to identifying and correcting mistakes in the initial stages of this project.
\end{acks}

\bibliographystyle{ACM-Reference-Format}
\bibliography{paper}

\end{document}